\documentclass[journal]{IEEEtran}

\usepackage{hyperref}
\usepackage{amsmath,amsthm,amssymb}
\usepackage{graphics, subfigure}
\usepackage{cite}
\usepackage{tikz}
\usepackage{mathtools}
\usetikzlibrary{shapes,arrows}
\usepackage{algorithm, algorithmic}
\usepackage[outdir=./]{epstopdf}
\usepackage{color}
\usepackage{multirow}
\usetikzlibrary{automata,positioning}
\usepackage{wrapfig}
\usepackage{rotating}

\usepackage{bbm}
\usepackage{url}
\usepackage{bbm}

\usepackage{enumitem}

\newtheorem{assumption}{Assumption}
\newtheorem{assum}{Assumption}
\newtheorem{theorem}{Theorem}
\newtheorem{lemma}{Lemma}

\newtheorem{definition}{Definition}
\newtheorem{remark}{Remark}
\newtheorem{cor}{Corollary}

\newtheorem{lem}{Lemma}
\newtheorem{example}{Example}

\title{Traffic Control in a Mixed Autonomy Scenario at Urban Intersections: An Optimization-based Framework}

\begin{document}
\author{Arnob Ghosh and Thomas Parisini
\thanks{This work has been partially supported by the European Union's Horizon 2020 research and innovation programme under grant agreement No 739551 (KIOS CoE) and by the Italian Ministry for Research in the framework of the 2017 Program for Research Projects of National Interest (PRIN), Grant no. 2017YKXYXJ.}
\thanks{A. Ghosh is with the Dept. of Electrical and Electronic Engineering at Imperial College London, UK {\tt\small (arnob.ghosh@imperial.ac.uk)}}
\thanks{T. Parisini is with the Dept. of Electrical and Electronic Engineering at Imperial College London, UK, 
with the KIOS Research and Innovation Centre of Excellence, University of Cyprus,
and also with the Dept. of Engineering and Architecture at the University of Trieste, Italy {\tt\small (t.parisini@gmail.com)}}
}
\maketitle

\begin{abstract}
We consider an intersection zone where autonomous vehicles (AVs) and human-driven vehicles (HDVs) can be present. As a new vehicle arrives, the traffic controller needs to decide and impose an optimal sequence of the vehicles which will exit the intersection zone. The traffic controller can send information regarding the time at which an AV can cross the intersection; however, the traffic controller can not communicate with the HDVs, rather the HDVs can only be controlled using the traffic lights. We formulate the problem as an integer constrained non-linear optimization problem where the traffic-intersection controller only communicates with a subset of the AVs. Since the number of possible combinations increases exponentially with the number of vehicles in the system, we relax the original problem and proposes an algorithm which gives the optimal solution of the relaxed problem and yet only scales linearly with the number of vehicles in the system. Numerical evaluation shows that our algorithm outperforms the First-In-First-Out (FIFO) algorithm.  
\end{abstract}

\section{Introduction}
\subsection{Motivation}
The autonomous vehicle technology has the potential to transform the transportation system in a disruptive manner. For example, the traffic intersection controller can communicate with the autonomous vehicles (AVs) \footnote{Here, we assume that the AVs are connected.}. \begin{color}{black} The traffic controller can obtain the optimal trajectory for each vehicle by solving an optimization problem  and subsequently, can communicate the optimal decision to the AVs. The AVs would then follow the solution provided by the traffic-controller. However, in initial years, the human driven vehicles (HDVs) and AVs will co-exist. While the intersection controller can communicate with the AVs, it can not communicate with the HDVs. Thus, such a mechanism will not work when both the AVs and the HDVs co-exist.   Thus, the traffic-intersection controller needs to decide the optimal schedule in which they will enter the intersection and control the traffic-light accordingly. The HDVs behave in a non-linear, complicated, and unpredictable manner. Thus, computing an optimal schedule is inherently challenging in a mixed autonomy. We seek to contribute in this space.\end{color}

\subsection{Related Literature}
Designing of mechanism for AVs to cross the intersection has been studied. In~\cite{dresner},  a reservation-based approach for intersection management is considered. \cite{li,yan,zhu} have focused on coordinating with vehicles to reduce the traffic delays. \cite{karaman} has proposed a polling mechanism to determine a sequence of times assigned to each vehicle. \cite{gilbert,hooker} considered minimizing the energy consumption for AVs. \cite{LEVIN2017528} developed a conflict point formulation to control the AVs. \cite{7963717,8690681} developed a MILP formulation to find optimal schedule for the AVs to cross a traffic intersection. \cite{yu2018integrated} considered an optimization based framework when there are only AVs.  A detailed discussion of the overall research in
this area can be found in \cite{rios-torres}.

\begin{color}{black}
Prof Cassandras' group has led several research efforts in developing optimal control mechanism to control the vehicles' trajectories at the intersections when there are only connected AVs. For example, \cite{cassandras1} has considered a decentralized optimal control framework for coordinating AVs crossing an urban intersection without explicit traffic signaling. In \cite{cassandras2}, a First-In-First-Out (FIFO) queuing structure has been proposed which specifies the crossing time for each AV. The authors then solved an optimal energy minimization problem to control the velocity and acceleration of each AV. In \cite{cassandras3}, the authors proposed a re-sequencing process by relaxing the FIFO order in order to increase the throughput. The idea is that if certain lanes in an intersection have less traffic, vehicles which arrive later to the intersection can move without waiting for vehicles in other lanes to cross the intersection if the safety constraint permits. \cite{zohdy} presented an approach based on Cooperative Adaptive Cruise Control for minimizing intersection delay and maximizing the throughput. Unsignalized traffic-intersection has also been also considered in \cite{zhang2020virtual,viana2019cooperative}.

Unlike the above papers, we consider a mixed traffic scenario consisting of HDVs and AVs.   In the near future, the AV and the HDV will co-exist. Hence, a traffic controller at the urban intersection should consider both the AVs and HDVs while taking its decision. The presence of the HDVs along with the AVs make the decision process more challenging.  First, the traffic-controller can not communicate with the HDVs. Second, one still requires  signalized traffic intersections to control the movement of the HDVs in order to avoid collision.  \cite{Zhang_2018} studied the impact of the HDVs on the energy consumption in an urban intersection. \cite{li2020,zheng,wu,guler2014using} studied the optimal arrangement of HDV and AV in a ring which would stabilize the system. \cite{niroumand2020joint} has considered introducing a white-phase inside the traffic-light phase to allow platooning. However, the above paper considered connected HDVs where HDVs can follow command provided by the traffic-intersection controller. Instead, we propose a mechanism where the HDVs need not be connected. Thus, our mechanism can be applied to a diverse set-up. Further, the above paper considered a cycle-based traffic-light system where each cycle is of fixed duration, instead, in our setting the traffic-light  is adapted based on the schedule of the vehicle, hence, the duration can be of variable length. 

The closest to our work is the recent optimization-based framework considered in \cite{hajbabaie}. The authors proposed an optimization-based framework to decide jointly the optimal trajectory of {\em all} the connected AVs (CAVs), and whether the next phase will be red or green at each lane. In contrast, our approach seeks to decide the optimal schedule of vehicles. Hence, instead of deciding whether the phases will be green or red, in our approach, the traffic-lights can be adjusted to much granular level as the traffic-lights can be changed based on the schedule of the vehicles which will enter the intersection.  Further, our approach does not seek to compute and inform the trajectory of all the CAVs, rather the traffic-controller only  provides information to those AVs which will face red-light when they become the lead vehicles.  Our proposed algorithm only informs the time at which an AV can enter. Once a schedule is computed, the traffic-intersection controller adjusts the traffic-light duration based on the schedule in our approach whereas in\cite{hajbabaie}, it needs to compute the trajectory every time instance to communicate with the AVs. Hence, it greatly reduces the communication cost compared to \cite{hajbabaie}. Our method only needs to recompute a schedule when a new vehicle enters. Thus, the computation cost is also smaller in our approach compared to \cite{hajbabaie}.  Our method is a combination of the decentralized and centralized approach--- where the traffic-controller only computes the optimal schedule decision centrally, the dynamics decisions are taken by the vehicles in a decentralized manner unlike the complete centralized approach considered in \cite{hajbabaie}. Our objective considers the fuel cost, and the velocity of the vehicles in addition to the delays of the individual vehicles unlike in \cite{hajbabaie}. \end{color}


\subsection{A Glimpse on the Optimization Problem}
We consider a traffic intersection zone at an urban setting. The intersection zone consists of a control zone and the merging zone, which together is called a {\em system} (Section~\ref{sec:system_model}).  \begin{color}{black}As a new vehicle arrives at the control zone, the traffic controller needs to determine the order at which the vehicles including the newly arrived vehicle would enter the intersection.  It is an online decision and the sequence is updated as soon as a new vehicle arrives. The traffic-light state at lanes is computed from the scheduling decision as the vehicles from the conflicting lanes can not enter the intersection simultaneously (Section~\ref{sec:traffic_light}). The dynamics of the vehicles depend on the traffic-light state. The traffic intersection controller informs  the times to an AV at which it can enter the intersection if it has entered the lane at a lead vehicle (i.e, the front vehicle at the control zone of that lane). If an AV is not the lead vehicle and can not enter the intersection immediately following its preceding vehicle, i.e., it faces a red-light when it becomes the lead-vehicle, the traffic controller also informs the AV at what time it can enter the intersection. The AV then adjusts its dynamics and enters the intersection at the specified time with the highest possible velocity (Section~\ref{sec:avlead}). When an AV  enters the intersection immediately after its preceding vehicle in the computed sequence, its dynamics is driven by the vehicles preceding the AVs (Section~\ref{sec:avfollow}). 

We assume that if a HDV is a lead vehicle, it  would decelerate at a uniform rate after seeing the red light (Section~\ref{sec:hdvlead}). When the green-light is switched on, if the HDV is the lead vehicle it would accelerate at a uniform rate till the velocity reaches the maximum value. The traffic controller can not send communicate with  the HDVs specifying the times at which they can enter the intersection unlike the AVs. Similar to the AV, we assume that when the HDVs follow other vehicles, their dynamics are governed by the preceding vehicles (Section~\ref{sec:avfollow}). The time at which all the vehicles enter the intersection and exit the intersection are completely specified once the lead vehicle's dynamics is specified at a lane using a recursive equation and the scheduling decision of the traffic-intersection controller. \end{color} 

We formulate the problem of determining the optimal sequence of vehicles crossing the intersection as an optimization problem (Section~\ref{sec:opti}). We consider the following objectives for the traffic controller: (a) the traffic controller tries to minimize the total energy cost. The energy cost is defined as the $L_2$-norm of the acceleration/retardation across the time the vehicles spend  in the system. (b) The traffic controller tries to minimize the total time each vehicle takes to cross the intersection. (c) The intersection controller seeks to maximize the velocity at which the vehicle can enter the intersection. If more vehicles enter the intersection at the speed limit it would increase the throughput. (d) The traffic controller seeks to minimize the overall time taken for all the vehicles to cross the intersection.\footnote{\begin{color}{black}Objectives (b) and (d) are not the same, (b) seeks to minimize the time each vehicle spends in the system; on the other hand, (d) seeks to minimize the total time it takes for all the vehicles to exit the intersection. \end{color}} In the optimization problems velocities and accelerations of the vehicles are constrained. The decision variable is the schedule of the vehicles present at the system, the scheduling decision must satisfy the safety constraint where the vehicles from the conflicting lanes can not present at the merging zone at a time (Section~\ref{sec:safety}). 

The optimization problem is computationally hard to solve.  The objective is non-linear and the equality constraints are non-linear since the time at which vehicles enter the intersection are non-linear functions of the dynamics. Further, the optimization problem has integer valued variables which specifies the order at which the vehicles would enter the intersection from the conflicting lanes. The number of possible sequences increase exponentially with the number of vehicles in the system (Section~\ref{sec:relax}). Thus, we consider a relaxed version of the problem where the traffic controller would only determine among the sequences involving the newly arrived vehicle and the vehicles in the conflicting lanes of the newly entered vehicle (Section~\ref{sec:relax}). The relative order of other vehicles in which they will enter the intersection is kept unchanged.  We show that such a relaxation drastically reduces the number of possible combinations  as the complexity only increases linearly with the number of vehicles in the worst case. We propose an algorithm which solves the relaxed optimization problem (Section~\ref{sec:algo}). 

\begin{color}{black}
We, empirically show that our algorithm reduces the total time taken by all the vehicles by at least 50\% compared to the FIFO sequence proposed in the existing literature (Section~\ref{sec:numerical}). In the FIFO, whichever vehicle enters the control zone first also enters the intersection first. We show that our algorithm significantly improves the performance when the penetration of the AV is small compared to the FIFO. We show that our algorithm increases the average velocity at which vehicles enter the intersection by at least 30\%, minimizes the time taken by each vehicle to cross the intersection by at least 40\%, and minimizes the magnitude of acceleration/deceleration compared to the FIFO. Unsurprisingly, we show that as the AV penetration rate increases, our algorithm performs better. We demonstrate that our proposed mechanism can be extended to the scenario where the HDVs deviate from the nominal values of the model assumed in the paper. We  empirically investigate the impact of the order of AV and HDV entering the system (i.e., the control zone). Empirical evidence suggest that the arrangement of vehicles where AVs follow HDVs rather the arrangement where HDVs follow AVs reduces the total time to cross the intersection. However, if AVs follow too many HDVs, the total time increases. Empirical evidence also suggests that the combination where the AV and HDV alternates can increase the average velocity of the vehicles at the intersection. \end{color}

\section{Problem Formulation}\label{sec:system_model}
In this section, we, first, describe the traffic intersection system we consider (Section~\ref{sec:zone}). We describe the safety constraint the traffic-intersection controller must satisfy while taking the decision (Section~\ref{sec:safety}). We describe how the traffic-intersection controller changes the red and green lights at a lane based on the optimal sequence (Section~\ref{sec:traffic_light}).  Subsequently, we describe the dynamics of AVs and HDVs (Sections~\ref{sec:avlead} and \ref{sec:hdv}) which describes how the vehicles approach the intersection for a given sequence of the vehicles in which they would enter the intersection. 

\subsection{The Intersection and Control Zone}\label{sec:zone}
We consider a merging zone in an urban setting. Let $N(t)$ be the cumulative number of vehicles which are present the system at time $t$.  The intersection system consists of a control zone and a merging zone. When a new vehicle arrives in the control zone, the traffic controller computes a sequence which denotes the order in which the existing vehicles will enter and exit the intersection system or the merging zone. 

The vehicles form a queue at different lanes at an intersection. Note that if a queue is non-empty, a vehicle must wait until at least all of the vehicles in-front of it enter the merging zone before it enters the merging zone. The control zone and the merging zone are depicted in the figure~\ref{fig:zone}.  The control zone is of length $L$ in each lane. After the control zone, a vehicle enters an intersection or merging zone of length $S$. 
\begin{figure}
\begin{center}
\includegraphics[width=90mm, height=70mm]{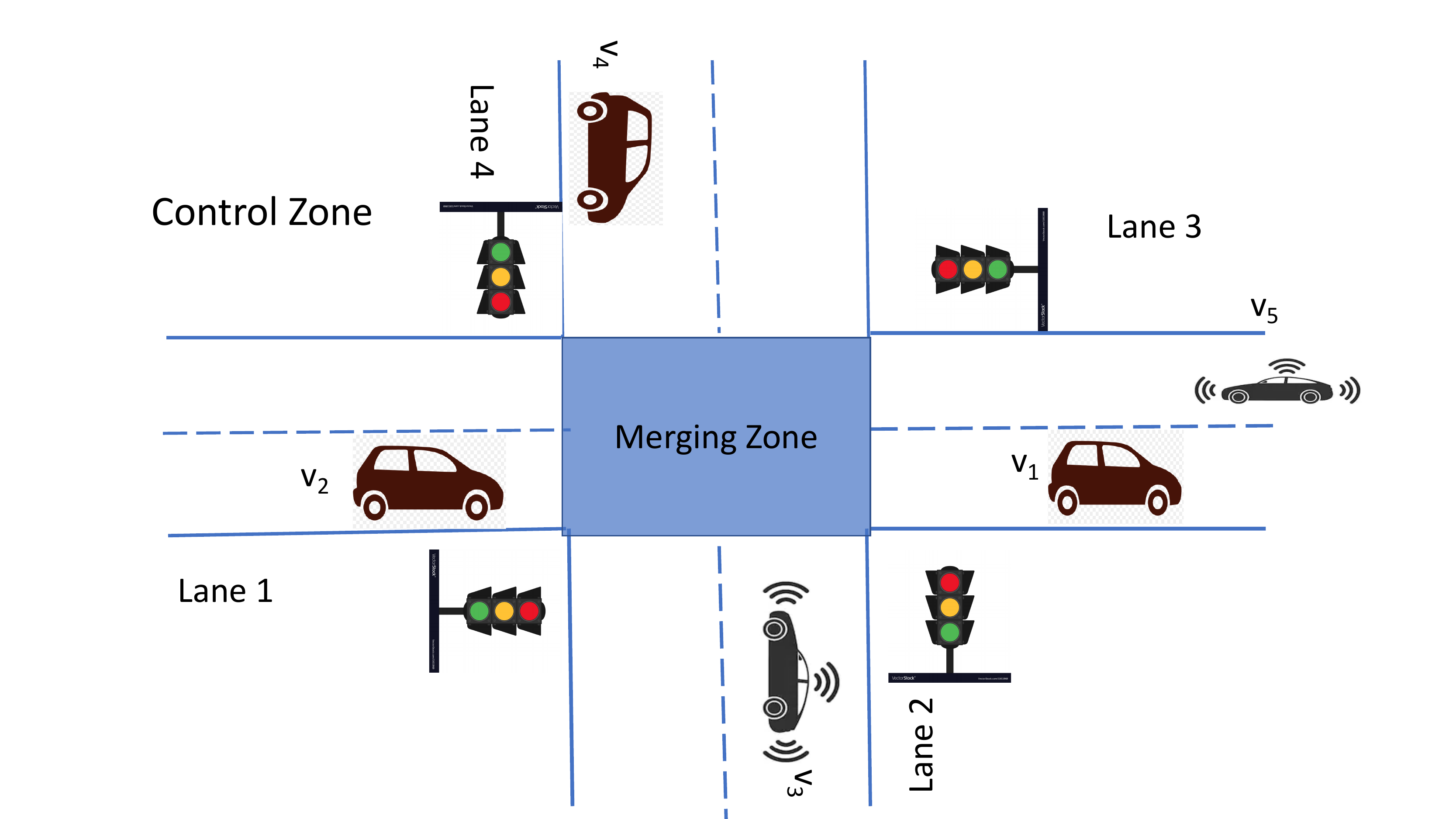}
\caption{The figure illustrates a traffic intersection control where a traffic controller decides the optimal sequence of vehicles which would cross the intersection and take actions accordingly. AVs are denoted with signs of being capable to communicate with the traffic controller. Vehicles from lane 1 can collide with vehicles from lane 2 and lane 4. Vehicles from lane 2 can collide with vehicles from lane 3 and lane 1.}
\label{fig:zone}
\end{center}
\vspace{-0.3in}
\end{figure}


Fig.~\ref{fig:zone} depicts that the scenario where  vehicle $V_5$ enters the system with $3$ more vehicles ($V_2, V_3$ and $V_4$) are already in the system. The traffic controller needs to determine the optimal exit sequence among the existing vehicles  in order to minimize the total travel time, maximize the velocity they can travel, and minimize the frequent change in velocities of the vehicles. Note that the vehicles from lane $3$ and vehicles from lanes $2$, and $4$ can not cross the intersection simultaneously. Hence, the traffic controller needs to determine vehicles from which lane should enter while closing the other lanes in order to avoid collision.


The traffic controller can coordinate with the AVs by informing the AVs know the time at which they can cross the intersection if they are the lead-vehicles (Definition~\ref{defn:lead}) or when they {\em can not enter} the intersection immediately after their preceding vehicles as the vehicles from other lanes will enter the intersection before. The traffic controller can  not coordinate with the HDVs. Rather, the HDVs only follow the traffic-lights.  If a traffic controller decides to stop a vehicle at a lane in order to allow vehicles from other lanes to enter the intersection it would impact the travel times of all the vehicles following the vehicle. 

\begin{color}{black}
Throughout this paper, the following assumption is in place.
\end{color}

\begin{assumption}
\begin{itemize}
\item The vehicles are going straight and no turn is allowed. The vehicles are either going from east to west and west to east, or north to south and south to north. 
\item The  traffic-controller can not communicate with the HDV regarding the time at which they can enter the intersection.  
\item The traffic intersection controller can sense the HDV's position and velocity accurately.
\item The AVs are connected and controllable. The traffic controller can send a signal when to enter the merging zone to a AV.
\item Lane changing is not acceptable within the control and intersection zone.
\end{itemize}
\end{assumption}
As described in Fig.~\ref{fig:zone}, there are 4 conflict points even for this basic set-up. Note that in practice, the traffic intersection controller can only have an estimate of a vehicle's position and velocity. In that case, we consider the worst case position and velocity of a vehicle. Empirically, we show the impact of the estimation errors in Section~\ref{sec:sensitivity}.  \begin{color}{black} Note that even though we have not considered any turning behavior, our model can be extended to accommodate the lane changing behavior as discussed in Section~\ref{sec:turn}.\end{color}

\subsection{Notations and Significance}
We introduce some notations which we use throughout the paper.  At the control zone, a vehicle is completely specified by the tuple $(i,j)$ where  $j$ denotes the lane and $i$ denotes the $i$-th vehicle in the lane from the intersection. The vehicles are numbered in the order they are closer to the intersection. Vehicle $(i-1,j)$ is the vehicle immediately preceding the vehicle $i$. Vehicle $i$ at lane $j$ enters the control zone at time $t_{i,j}^{0}$. 

The position and speed of the vehicle $(i,j)$ at time $t$ are denoted by $p_{i,j}(t)$ and $v_{i,j}(t)$ respectively.
\begin{color}{black}
$v_{m}$  is the speed limit. $v_{free}$ is the desired velocity of the system. In our model, we consider $v_{free}=v_m$. However, in the congested setting, for example, if there is a congestion in the neighboring intersections then $v_{free}$ can be smaller. Our model accommodate any velocity $v_{free}$ as discussed in Section~\ref{sec:free_velocity}. A vehicle enters the intersection at a speed $v_m$ if there is no congestion, otherwise, it may enter the intersection at a different speed.  \end{color}

\begin{color}{black}
\begin{definition}
For vehicle of any type, the time at which vehicle $(i,j)$ would cover the distance $L+S$ starting from time $t_{i,j}^{0}$ and travelling at a constant speed $v_m$ is given by $t_{i,j}^{\prime}$:
\begin{align}\label{tprime}
t_{i,j}^{\prime}=t_{i,j}^{0}+\dfrac{L+S}{v_{m}}.
\end{align}
The time at which the vehicle $(i,j)$ would cover the distance $L$ starting from $t_{i,j}^{0}$ at a constant speed $v_m$ is
\begin{align}
t_{i,j,m}^{\prime}=t_{i,j}^{0}+\dfrac{L}{v_{m}}.
\end{align}
\end{definition}

\noindent
$t_{i,j,m}^{\prime}$ and $t_{i,j}^{\prime}$ indicate the times vehicle $(i,j)$ would reach the intersection and cross the intersection if it would not be subject to any traffic. Note that $t_{i,j}^{\prime}$ ($t_{i,j,m}^{\prime}$, resp.)  is equal to the initial time plus the time the vehicle takes to travel the distance $L+S$ ($L$, resp.) at the maximum speed. Thus, they represent the times the vehicle take to exit the control and merging zones in the idealistic scenario where the vehicle would exit the system if there is no traffic.  These values will be used to compute the total delay the vehicle $(i,j)$ would face because of the traffic and the decision of the traffic-intersection controller. 
\end{color}

When the vehicle $(i,j)$ enters the control zone, its position $p_{i,j}(t_{i,j}^{0})=0$. When the vehicle $(i,j)$ in the control zone enters the intersection or merging zone then it is denoted as $(i^{\prime}_i,j)$ which explicitly denotes the identity of the  at the control zone. Once a vehicle exits the intersection, that vehicle is removed from the system. 
\begin{color}{black}
\begin{definition}
We denote $t_{i,j}^{m}$ and $t_{i,j}^f$ as the times the vehicle $(i,j)$ enters and exits the intersection respectively. 

Formally, $t_{i,j}^m=\inf\{t|p_{i,j}(t)>L\}$, and $t_{i,j}^{f}=\inf\{t|p_{i,j}(t)=L+S\}$. 
\end{definition}
Note that in the idealistic scenario, $t_{i,j}^{m}=t_{i,j,m}$ and $t_{i,j}^f=t_{i,j}^{\prime}$.  However, because of the traffic and the decision of the traffic-intersection controller those relationships may hold in our scenario. 
\begin{definition}\label{def:delay}
We denote delay as $t_{i,j}^f-t_{i,j}^{\prime}$ since this is the additional time the vehicle $(i,j)$ takes to exit the system because of the traffic and the traffic-intersection controller's decision.
\end{definition}

\begin{definition}\label{defn:lead}
When $i=1$ for vehicle $(i,j)$, then it is the lead-vehicle in the control zone of lane $j$.
\end{definition}
\end{color}
\begin{definition}\label{defn:lj}
Let $\mathcal{L}_j$ denote the set of vehicles in a lane $j$.  The set $\mathcal{L}_j$ is ordered in the sequence the vehicles arrive in the lane $j$. 
\end{definition}
The first element in the set $\mathcal{L}_j$ is the first one among the existing vehicles in lane $j$ to arrive in the lane $j$ and it is the lead vehicle in lane $j$. 

\begin{definition}\label{defn:mj}
Let $\mathcal{M}_j$ denote the set of conflicting lanes of lane $j$. 
\end{definition}
In Fig.~\ref{fig:zone}, lanes 2 and 4 belong to $\mathcal{M}_1$. Since we have $4$ lanes, for the non-conflicting lane $j^{\prime}$ of $j$, $M_{j^{\prime}}=M_j$. For example, in lanes 2 and 4 also belong to $\mathcal{M}_{j^{\prime}}$.

\begin{definition}
Let $b_{i,j,l,k}=0$ if the vehicle $(i,j)$ enters the intersection before the vehicle $(l,k)$ where $k\in \mathcal \mathcal{M}_j$ and $l\in \mathcal{M}_k$, otherwise $b_{i,j,l,k}=1$.
\end{definition}
If $b_{i,j,l,k}=0$, then $b_{l,k,i,j}=1$. Note that the traffic-controller decides $b_{i,j,l,k}$ which characterizes the sequence order. Since no passing is allowed within a lane, it is obvious that vehicle $(i,j)$ can only enter the intersection after $(i-1,j)$ if $i>1$. We only need to compute sequence order among the vehicles at the conflicting lanes. At time $t$, the traffic-intersection controller decides $b_{i,j,l,k}$ to schedule the vehicles among the conflicting lanes by maintaining the safety constraint which we describe next. 

\subsection{Safety Constraint}\label{sec:safety}
The traffic controller needs to ensure that the vehicles from colliding lanes can not be present at the intersection at the same time. Thus, vehicle $(i,j)$ and any vehicle from the lanes $\mathcal{M}_j$ can not be present simultaneously at the intersection.  For example, in Fig.~\ref{fig:zone} the vehicles from lane $1$ and $2$ can not cross the intersection at the same time since they can collide with each other. Thus, lanes $1$ and $2$ are colliding lanes with each other. Vehicles from a lane can only enter the merging zone when the there are no vehicle from the colliding lane in the merging zone. 


Recall that $t_{i,j}^{m}$ is the time when the vehicle $(i,j)$ enters the intersection. Hence, the safety or feasibility constraint for any vehicle $i\in \mathcal{L}_j$ is 
\begin{align}\label{eq:safety}
(t_{i,j}^{m},t_{i,j}^{f})\cap (t_{l,k}^{m},t_{l,k}^{f})=\Phi \quad \forall l\in \mathcal{L}_k,\forall k\in \mathcal{M}_j
\end{align}
The above constraint ensures that any vehicle in the conflicting lanes can not enter the intersection at the same time with vehicle $(i,j)$. Note that certainly the constraint can be relaxed and vehicles from conflicting lanes $\mathcal{M}_j$ can enter at times $t_{i,j}^{f}-\epsilon$ where $\epsilon$ is small enough to ensure that the vehicles will not collide. Our analysis and model can be easily extended to the above scenario. 

 (\ref{eq:safety}) can be represented in terms of the decision variable $b_{i,j,l,k}$
\begin{align}\label{order}
& t_{i,j}^{m}\geq b_{i,j,l,k}t_{l,k}^f \quad \forall l\in \mathcal{N}_{j,k},\forall k\in \mathcal{M}_j\nonumber\\
& t_{l,k}^{m}\geq (1-b_{i,j,l,k})t_{i,j}^f
\end{align}
If $b_{i,j,l,k}=0$, then, the vehicle $(i,j)$ would enter the intersection before the vehicle $(l,k)$ and the vehicle $(l,k)$ can only enter the intersection at time $t_{i,j}^f$, the time at which the vehicle $(i,j)$ exits the intersection. The traffic-controller decides $b_{i,j,l,k}$ which in turn specifies the red-light and green-light timings at a lane, which in turn governs the dynamics of the vehicles and the times $t_{i,j}^m$ and $t_{i,j}^f$. 
\begin{color}{black}
\subsection{Traffic Light System}\label{sec:traffic_light}
In this section, we describe how the decision $b_{i,j,l,k}$ characterizes the time when the traffic-light is red or green at each lane. The traffic controller computes an order of vehicles in which they will enter the intersection or merging zone. When a vehicle enters the merging zone, the corresponding traffic-light at that lane must be green.

If the traffic-intersection controller decides $b_{i,j,l,k}=1$, and $b_{i-1,j,l,k}=0$,i.e., the vehicle $(l,k)$ from conflicting lane $k$ enters the intersection before vehicle $(i,j)$ and after vehicle $(i-1,j)$ enters the intersection, then the traffic controller will put a red-light at lane $j$  after the vehicle $(i-1,j)$ enters the intersection. We now describe when the vehicle $(i,j)$ will face green-light again.  

First, we introduce some notations. 
\begin{definition}\label{defn:ji}
  Let  $j_i(k)=\arg\max_{l\in \mathcal{L}_k, k\in \mathcal{M}_j}\{l|b_{i,j,l,k}=1\}$
\end{definition}
$j_i(k)$ denotes the largest index among the vehicles $(l,k)$ in the conflicting lane $k$ of the lane $j$ which will enter the merging zone before the vehicle $(i,j)$. Vehicles $(1,k),\ldots,(j_i(k),k)$ will enter the intersection before the vehicle $(i,j)$. Hence, the vehicle $(j_i(k)+1,k)$ (if it exists) can only enter the merging zone after the vehicle $(i,j)$ exits the intersection (otherwise $j_i(k)+1$ would have been the maximum). 

Note that if $j_{i-1}(k)$ and $j_i(k)$ are not the same for any $k\in \mathcal{M}_j$, then vehicle $(i,j)$ can not enter the intersection immediately preceding the vehicle $(i-1,j)$. Rather vehicles $((j_{i-1}(k)+1,k),\ldots, (j_i(k),k))$ from the conflicting lane $k$ will enter the intersection before the vehicle $(i,j)$. Thus, in between a red-light must be on at lane $j$. 

Let $\phi^r_j(t)$ and $\phi^g_j(t)$ denote whether the traffic-light is red and green respectively at lane $j, j=1,\ldots,4$ at time $t$. Specifically, if $\phi^r_j(t)=1$, the traffic-light is red at lane $j$, otherwise it is $0$. On the other hand, $\phi^g_j(t)=1$ denotes that the traffic-light is green at lane $j$ at time $t$. \begin{color}{black}We do not consider any amber-light and we denote how to accommodate the amber light in Section~\ref{sec:amber}. \end{color}

If $i\neq 1$ when it enters the control zone, then, it will be the lead-vehicle after $t_{i-1,j}^m$, i.e, when the vehicle $(i-1,j)$ enters the intersection. If $j_{i-1}(k)$ and $j_i(k)$ are different for any $k\in \mathcal{M}_j$, then after the vehicle $(i-1,j)$ enters the intersection, the vehicle $(i,j)$ has to wait till the time $(j_i(k),k)$ exits the intersection. Note that there are two conflicting lanes of lane $j$ since we consider a 4-way intersection. Hence, we also need to consider $b_{i,j,l,k_1}$ where $k_1\in \mathcal{M}_j$, and $k_1\neq k$.

If $j_i(k)$ and $j_{i-1}(k)$ are different for any $k\in \mathcal{M}_j$, $\phi^g_j(t_{i-1,j}^m)=1$, however, $\phi^r_j(t_{i-1,j}^m+)=1$, i.e., the traffic-light is red at lane $j$ as soon as the vehicle $(i-1,j)$ enters the vehicle. Hence, the traffic-light can be adapted at any lane. The green-light will be switched on at lane $j$ again at time
\begin{align}\label{eq:t_gij}
t^g_{i,j}=\max_{k\in \mathcal{M}_j}t_{j_i(k),k}^{f}
\end{align}
Hence, $\phi^r_j(t)=1$ for $t\in (t_{i-1,j}^{m},t^g_{i,j})$. 
Specifically, the vehicle $(i,j)$ has to wait till the time all the vehicles $(l,k)$ such that $b_{i,j,l,k}=1$ exit the intersection. If $j_i(k)$ and $j_{i-1}(k)$ are the same for all $k\in \mathcal{M}_j$, then the vehicle $(i,j)$ will face the green-light after the vehicle $(i-1,j)$ enters the merging zone. In this case $\phi^r_j(t)=0$ for $t\in (t_{i-1,j}^m,t_{i,j}^m)$. Thus, the vehicle $(i,j)$ does not observe the red-light when it becomes the red-light. 

The decision $b_{i,j,l,k}$ completely characterizes the times at which the red-light or green-light will be switched on at lanes $j$. For example, first see among the front vehicles, i.e, $i=1$ at lane $j$ and $l=1$ at lane $k\in \mathcal{M}_j$. If $b_{1,j,1,k}=0$, for all $k\in \mathcal{M}_j$, then the green-light will be on for lane $j$ if $j_2(k)$ and $j_1(k)$ are the same then green-light will be on till the time at least vehicle $(2,j)$ enters the intersection at lane $j$. If $j_i(k)$ is different for some $i$ at lane $j$ compared to $j_{i-1}(k)$, the red-light will be switched on after the vehicle $(i,j)$ enters the intersection. The vehicles till $(j_i(k),k)$  for $k\in \mathcal{M}_j$ will then exit the intersection before again the green-light will be switched on at lane $j$. 

On the other hand, if $b_{i,j,l,k}=0$ for some $k\in \mathcal{M}_j$, then, the vehicle $(l,k)$ would enter the intersection before the vehicle $(i,j)$ where $i=1$. The time the green-light will be on at lane $j$ is computed from $t^g_{1,j}$. 

Finally, the times at which the vehicles enter and exit the intersection depend on the red-light and green-lights which in turn depend on the   decision $b_{i,j,k,l}$. Thus, we slight abuse of notation, we denote $t_{i,j}^{m}(\textbf{b})$ and $t_{i,j}^f(\text{b})$ as functions of $\textbf{b}=\{b_{i,j,k,l}\}$. We sometimes make this dependence explicit. The dependence are explicitly defined by the dynamics of the vehicles and $b_{i,j,l,k}$ which we describe next. 
\begin{assumption}\label{assum2}
If there is no vehicle in the control zone at a lane, the red-light will be switched on.
\end{assumption}
The above assumption is made for the ease of exposition and can be relaxed. The above assumption states that when the vehicle $(i,j)$ is the lead-vehicle at the control zone as it enters the lane $j$, it will initially face a red-light. The green-light will be switched on based on the decision made by the traffic-controller. 

Note that when $j_i(k)$ and $j_{i-1}(k)$ are different, then the vehicle $(j_i(k),k)$ would enter the intersection before the vehicle $(i,j)$. However, the vehicle $(j_i(k)+1,k)$ would not enter the intersection before the vehicle $(i,j)$. 
Hence, the red light will be switched on at lane $k$ when the vehicle $(j_i(k),k)$ enters the intersection. Thus, in the time interval $(t^m_{j_i(k),k},t^g_{i,j})$ all the traffic-lights are red. This ensures that the intersection is devoid of any vehicles from the conflicting lane $k$ before the traffic-light is made green at lane $j$. Whenever a schedule contains two consecutive vehicles from conflicting lanes, red-light will be on for all the lanes in order to make sure that the vehicle in the preceding order exits the intersection. 

\end{color}

\subsection{Dynamics of AV and communication with the traffic controller in the control zone}\label{sec:avlead}

In this section, we describe the dynamics of the vehicles impacted by the decisions of the traffic-intersection controller. We first describe the dynamics for the AV $(i,j)$ where $j_i(k)$ and $j_{i-1}(k)$ are different  (cf.Definition~\ref{defn:ji}) for some $k\in \mathcal{M}_j$. Here, we also discuss the scenario where the AV enters the lane as a lead-vehicle. The scenario where for the AV $(i,j)$ $j_i(k)$ and $j_{i-1}(k)$ are the same is considered in Section~\ref{sec:avfollow}.   

In the above scenario where $j_i(k)$ and $j_{i-1}(k)$ are different for some $k\in \mathcal{M}_j$ , the traffic-intersection controller informs the AV when it can enter the intersection by computing $t^g_{i,j}$ (cf. (\ref{eq:t_gij})). The traffic-controller also informs the time at which the vehicle $(i,j)$ can enter the intersection when the AV $(i,j)$ is the lead-vehicle, i.e., it is the only vehicle at the control zone when it enters. 

\noindent
AV $(i,j)$ is governed by the following dynamics in this scenario
\begin{align}\label{dyna}
& \dfrac{dp_{i,j}}{dt}=v_{i,j}(t),\quad p_{i,j}(t_{i,j}^0)=0, \nonumber\\& \dfrac{dv_{i,j}}{dt}=u_{i,j}(t), \quad v_{i,j}(t_{i,j}^0)=v_{ini}
\end{align}
Recall that $t_{i,j}^0$ is the time when vehicle $(i,j)$ enters the system. $t_{i,j}^{f}$ is the time the vehicle $i$ leaves the system. $u_{i,j}(t)$ is the control signal or acceleration for the AV $(i,j)$. 

To ensure that the speed and the control input of the AVs are within acceptable limits, we have the following
\begin{align}\label{constr}
u_{min}\leq u_{i,j}(t)\leq u_{max}\nonumber\\
v_{min}\leq v_{i,j}(t)\leq v_{max}
\end{align}
\begin{color}{black} Note that in our setting, $v_{max}$=$v_m$.\end{color}

At time $t$, the traffic-controller computes the decision $b_{i,j,l,k}$. If $j_i(k)$ and $j_{i-1}(k)$ are different, then the traffic-intersection controller computes the time $t^{g}_{i,j}$ at which it can enter the intersection and informs the AV. The AV then updates its dynamics at time $t$. 

Since the AV can not enter the intersection before the time $t_{i,j,m}^{\prime}$, the traffic controller would let the AV know a time $t_{i,j}^{m}\geq t_{i,j,m}^{\prime}$ at which it can enter the intersection. The AV would like to approach the intersection at the maximum velocity $v_{m}$. 
We assume that the AV would try to minimize the total energy expenditure. Hence, the AV controller solves the following problem
\begin{align}\label{eq:min}
\text{minimize }& \int_{\tau=0}^{t_{i,j}^m-t}u_{i,j}(t)^2dt\nonumber\\
\text{subject to }& p_{i,j}(t_{i,j}^m)= L,\quad v_{i,j}(t_{i,j}^m)=v_{m}\nonumber\\
& (\ref{dyna}), \& (\ref{constr})
\end{align}
The initial conditions are the velocity and position of the vehicle $(i,j)$ at time $t$ when the optimization problem is solved. The first objective specifies the cost due to acceleration or deceleration.  The first constraint ensures that the distance traversed by the vehicle $(i,j)$ is $L$ at $t_{i,j}^{m}$ and the second constraint ensures that when the AV enters the intersection, it does so at the free speed. Note that $t_{i,j}^m=\max\{t^g_{i,j},\tilde{t}_{i,j}\}$. 

On the other hand, since the maximum acceleration is bounded by $a_{max}$, thus, we also have
\begin{align}
    t_1=\dfrac{v_m-\bar{v}}{a_{max}}\nonumber\\
    d=\bar{v}*t_1+\dfrac{1}{2}a_{max}t_1^2\nonumber\\
    t_2=(L-d-\bar{p})/v_m
\end{align}
i.e, it would take $t_2$ time to reach the intersection at the maximum speed since the acceleration is upper bounded by $t_2$ from the current velocity $v_{i,j}(t)=\bar{v}$, and the current position $p_{i,j}(t)=\bar{p}$. Thus, $t_{i,j}^{m}=\max\{t^g_{i,j},\tilde{t}_{i,j},t+t_2\}$.

\begin{color}{black}
Even though we consider that the vehicle $(i,j)$ enters the intersection at the maximum speed, our model can be extended to the scenario where the traffic-controller mandates that the velocity at the intersection be velocity other than the speed-limit in order to control the speed of vehicles at the neighboring intersections. How the intersection-controller chooses such a velocity is beyond the scope of this paper. 
\end{color}
\begin{theorem}\label{thm:1}
If a feasible solution exists, the solution is unique,  and we denote the solution as
\begin{align}\label{eq:leadavdyn}
(u_{i,j}(t),v_{i,j}(t),p_{i,j}(t))=F(t_{i,j}^{0},t_{i,j}^{m},v_{initial},L)
\end{align}
where $v_{initial}=v_{free}$ is the initial velocity and $L$ is the distance need to be transversed from time $t_{i,j}^{0}$ to $t_{i,j}^{m}$. 
\end{theorem}
\begin{proof}Please see Appendix~\ref{proof}. 
 \end{proof}

When $v_{ini}=v_m$, and $t_{i,j}^{m}=t_{i,j,m}^{\prime}$ ,$u_{i,j}(t)=0$ for all $t$ from the optimization problem in (\ref{eq:min}). The vehicle $i$ in lane $j$ would move at uniform maximum speed. Note that the function $F(\cdot)$ specifies the control acceleration, velocity and position at every time of the AV when the time to enter the intersection and the initial times are specified. 

\begin{remark}
The energy cost may be smaller in the scenario where we do not put the requirement that the velocity of the AV when it would enter the intersection must do at the speed close to the speed limit. We can consider such a scenario by relaxing the constraint. However, detailed analysis is kept for the future. 
\end{remark}

\begin{definition}\label{defn:assign}
If for a certain decision $\mathbf{b}$, there does not exist any feasible solution of (\ref{eq:min}),then we denote $t_{i,j}^{m}(\mathbf{b})=\infty$.
\end{definition}
Thus, if for a certain decision taken by the traffic-intersection controller, there is no feasible solution of (\ref{eq:min}), we specify the time $t_{i,j}^m$ as $\infty$. We utilize the above to discard such sequence which renders the solution infeasible when we describe the optimization problem for the traffic-intersection controller in Section~\ref{sec:opti}. 

\subsection{Human-Driven Vehicle Dynamics}\label{sec:hdv}
The traffic-controller can not communicate the time the HDV can enter the intersection specifically when it becomes the lead-vehicle in the control zone. Rather, the dynamics of the HDV when it is the lead-vehicle at the control zone, is impacted whether the traffic-light is red or green at time $t$.


\subsubsection{The dynamics when a HDV is the leading vehicle}\label{sec:hdvlead}

\begin{color}{black}
Generally, a HDV only brakes when it is at most $l$ ($l\leq L$) distance from the intersection  when it sees a red light. Thus, even if the red-light is on when the HDV is more than $l$-distance away from the intersection, the HDV would not decelerate.   If $l=L$, it means that the control zone is small, hence, the vehicle $(i,j)$ would decelerate as soon as it enters the control zone if the traffic-light is red at lane $j$. \end{color}

We assume that the HDV will decelerate with uniform retardation $a$. Hence, the dynamics of the HDV, when it encounters the red light, i.e, $\phi^r_j(t)=1$, is the following
\begin{eqnarray}
\dfrac{dv_{i,j}(t)}{dt}=\begin{cases}-a\quad\text{if }p_{i,j}(t)>L-l,\\
0\quad \text{otherwise.}\end{cases}\label{leadh}
\end{eqnarray}
\begin{color}{black}
the deceleration $a$ is such that the vehicle stops just before entering the merging zone. In particular, if the vehicle $(i,j)$ faces the red-light (i.e, the red-light is switched on) at a distance $s\leq l$ away from the intersection, then the magnitude of $a$ is
$\dfrac{v_s^2}{2s}$
where $v_s$ is the velocity when it faces the red-light.\end{color}

\begin{color}{black}
Note that it may happen that when the vehicle $(i,j)$ faces the red-light, it is very close to entering the intersection, then, the deceleration may be below the minimum value $u_{min}$. Specifically, if
\begin{align}\label{eq:lowerumin}
    \dfrac{v_s^2}{2s}>|u_{min}|
\end{align}
then, a deceleration is required which is lower than the minimum value. Thus, such a sequence would not be feasible. Hence, similar to Definition~\ref{defn:assign}, we make $t_{i,j}^{m}=\infty.$ In Section~\ref{sec:opti}, we then see that such a sequence will never be chosen. 
\end{color}


When a HDV sees a green light again, it would accelerate unless it is moving at the maximum speed. We assume that the HDV accelerates with a constant acceleration $a_{h}$. Thus, when $\phi^g_j(t)=1$, we have
\begin{align}\label{eq:dynh_acc}
\dfrac{dv_{i,j}(t)}{dt}=\begin{cases} a_h\quad \text{if } v_{i,j}(t)<v_{m},\\
0\quad \text{otherwise}\end{cases}.
\end{align}
Note that the red-light and green-light phases are governed by the decision $b_{i,j,l,k}$ taken by the intersection controller. 

\begin{color}{black}
\subsubsection{Dynamics of vehicles when it follows other vehicles}\label{sec:avfollow}
The dynamics of HDVs when they follow the other vehicles and the AV $(i,j)$ for which $j_i(k)$ and $j_{i-1}(k)$ are the same (i.e., it will face green-light after it becomes the lead-vehicle), are given by popular vehicle following model. We consider the following model

The dynamics of a HDV in such a scenario is described using the Intelligent Driver Model (IDM) \cite{treiber2000congested}.
It is an easy-to-tune adaptive cruise control system able to avoid vehicles collision in car-following mode.
The dynamics for the vehicle $(i,j)$ is given by:
\begin{equation}
\label{eq:IDM_follower_green}
\dot{v}_{i,j}(t)= u_{max}\left( 1-\left(\frac{v_{i,j}(t)}{v_m}\right)^4-\frac{\left(s^*_{\{i,i-1\}}(t)\right)^2}{s^2_{\{i,i-1\}}(t)+\epsilon^2}\right),
\end{equation}
where  $s_{\{i,i-1\}}(t)= p_{i-1,j}(t)-p_{i,j}(t)$ is the current inter-vehicle distance, $\epsilon$ is a small number which makes the denominator always non-zero, and $s^*_{\{i,i-1\}}$, the desired inter-vehicle distance:
\begin{eqnarray}
 s^*(v_{i,j}(t),\Delta v_{\{i,i-1\}}(t)) = s_0 + T_{i,j}v_{i,j}(t) + \nonumber\\
 \frac{v_{i,j}(t)\Delta v_{\{i,i-1\}}(t)}{2\sqrt{u_{\text{max}}|u_{\text{min}}|}},
\label{eq:dec_green}
\end{eqnarray}
where $\Delta v_{\{i,i-1\}}(t)=v_{i,j}(t)-v_{i-1,j}(t)$ is the vehicles' difference in speed, $u_{\text{min}}$ is the maximum deceleration, $s_0$ is the safe distance, and $T_{i,j}$ the safety time gap between two vehicles. Note that $s^{*}$ is the ideal distance that should be maintained between the vehicles $(i,j)$ and $(i-1,j)$ when the vehicle $(i,j)$ is moving at a speed $v_{i,j}(t)$ and the inter-vehicular speed is $\Delta v_{\{i,i-1\}}(t)$. $s^{*}$ increases as the vehicle $(i,j)$ approaches a slower vehicle (since the vehicle may need to brake). On the other hand, $s^{*}$ decreases as the vehicle $(i,j)$ approaches a faster vehicle. 

In order to ensure that the velocity is bounded between $v_{min}$ and $v_{max}$, we impose
\begin{align}\label{eq:constraint}
    v_{i,j}(t)=\max\{\min\{v_{i,j,}(t),v_{max}\},v_{min}\}
\end{align}

If the vehicle $(i,j)$ is an AV, we use the same dynamics. However, the parameter $T_{i,j}$ is smaller if the vehicle is an AV. In Section~\ref{sec:deviation}, we discuss how our model can accommodate parameter values different from the estimated ones. Note that our model can easily be extended to other popular vehicle following models such as Newell\cite{NEWELL2002195}, optimal velocity model, and full velocity model. We use the IDM since empirically it performs the best among all the other vehicle models to mimic the behavior of the vehicles \cite{zhu2018modeling,kanagaraj2013evaluation}. 
\end{color}



\begin{remark}\label{rmk:vehicledyn}
Note that velocity and acceleration of a vehicle is completely specified by the lead vehicle dynamics and the decision $b_{i,j,l,k}$. When the lead vehicle dynamics is known, the vehicle following the lead vehicle is computed according to the dynamics of the vehicles. Thus, in a recursive manner, the dynamics of all the existing vehicles in a lane are known. 

The time at which the vehicle $(i,j)$ will enter and exit the intersection given by the position model of the vehicle. Thus, the time all the vehicles will enter and exit the intersection are also known in recursive manner. 
\end{remark}
Note that the values of $b_{i,j,l,k}$ precisely determine $\phi^g_j(t)$ and $\phi^r_j(t)$ for a lane $j=1,\ldots,4$ as described in Section~\ref{sec:traffic_light}. Thus, the dynamics of vehicles are completely characterized once $b_{i,j,l,k}$ is known. 

Hence, the time at which the vehicle $(i,j)$ enters the intersection or merging zone computed at time $t$ is given by 
\begin{align}\label{enter}
t^{m}_{i,j}=\psi^m_{i,j}(t,\mathbf{b})
\end{align}
here $\psi(\cdot)$ also depends on the positions, and velocities of vehicles at time $t$. As mentioned $t^{m}_{i,j}$ is computed recursively starting from the first vehicle in the control zone at lane $j$. 

We represent the velocity and position profile of the vehicle $(i,j)$ as the function of the decision variable as the following
\begin{align}\label{pos_vel}
(v_{i,j}(t),p_{i,j}(t))=\zeta_{i,j}(t,\mathbf{b})
\end{align}
The velocity and position profiles are also obtained in a recursive manner based on the decisions $b_{i,j,l,k}$.
 
 \begin{color}{black}
\subsection{Dynamics of the vehicles in the merging zone}\label{sec:merge}
When in the merging zone, the vehicle $(i,j)$ is represented by $(i^{\prime}_i,j)$. if the vehicle $(i^{\prime}_i,j)$ is leader at lane $j$, its dynamics is given by (\ref{eq:dynh_acc}). The length of the merging zone is such that the velocity can reach the maximum value $v_m$ within the merging zone. If the vehicle is following another vehicle then its dynamics is given by (\ref{eq:IDM_follower_green}) irrespective of whether the vehicles is AV or HDV. Note that when a vehicle enters the merging zone, it does not face any traffic light anymore, hence, its dynamics is not impacted by the traffic-light. When the traffic-controller informs the AV the time at which it can enter the intersection, the AV enters at the maximum speed as specified in Theorem 1. Also, note that such an AV will also be the leading vehicle. Hence, the AV moves at the constant speed after entering the merging zone. 

The decision $b_{i,j,l,k}$ specifies when the vehicle $(i,j)$ will exit the intersection which is computed using the velocity and position profile of the vehicle. We characterize
\begin{align}\label{eq:exit}
t^f_{i,j}=\psi^f_{i,j}(t,\mathbf{b})
\end{align}
$\psi^f_{i,j}$ is inherently a function of the velocity and position of vehicle $(i,j)$ at time $t$ and decisions $b_{i,j,l,k}$. Similar to $\psi^m_{i,j}$, $\psi^m_{i,j}$ can also be computed recursively starting from the leading vehicle and $b_{i,j,l,k}$. 

\end{color}

\section{Traffic Controller's Decision-Making Strategy}
When a new vehicle $(i,j)$ arrives, the traffic controller decides the sequence at which the vehicle will exit the intersection. In this section, we, first, describe a relaxed problem which significantly reduces the computational complexity (Section~\ref{sec:relax}). Subsequently, we formulate the optimization problem, the solution of which gives the optimal sequence (Section~\ref{sec:opti}). 

\begin{color}{black}
We introduce a notation which we use throughout this section. Let $Seq(t)$ be the  order of the vehicles present at time $t$ at the control zone which will enter the intersection. When a vehicle exits the control zone and enters the merging zone, the vehicle is removed from the control zone and put into the merging zone system. All the vehicles ids' at that lane are reduced by $1$, i.e., the vehicle $(i,j)$ becomes $(i-1,j)$ if a vehicle from lane $j$ enters the merging zone. Once a vehicle exits the intersection, that vehicle is removed from the system. The traffic-intersection controller only seeks to schedule vehicles at the control zone. \end{color}

\begin{color}{black}
\subsection{Reducing the Complexity}\label{sec:relax}

The possible number of combinations increases exponentially with the number of vehicles in the system. For example, if there are two conflicting lanes and there are $N$ number of vehicles in each lane, the possible number of combinations is $\binom{2N}{N}$ which is of the order of $N^N$. Further, the dynamics of the vehicles are non-linear functions of the decision. Hence, finding an optimal sequence is inherently computationally challenging. 

The traffic controller needs to decide the sequence order in real time, possibly within a few milliseconds. In order to reduce the complexity, we reduce the number of possible combinations by assuming the following:

\begin{assumption}\label{assum_reduce}
When a new vehicle $(i,j)$ (i.e., it arrives at lane $j$), we only consider the sequence order of the new vehicle  and the existing vehicles at the conflicting lanes $\mathcal{M}_j$. The relative order among the vehicles other than $(i,j)$ are kept the same. 
\end{assumption}
 
 The above assumption entails that if the sequence order before a new vehicle $(i,j)$ arrives be $S^{old}$, i.e,, $S^{old}=Seq(t_{i,j}^{0}-)$ where $t_{i,j}^{0}-$ is the time just before $t_{i,j}^{0}$, then the sequence order after the new vehicle $(i,j)$ arrives be $S^{new}=Seq(t_{i,j}^0)$ such that $S^{new}(m)=S^{old}(m)$ for any $m\leq m_1$ where$S^{old}(m_1)=(i-1,j)$. Thus, the order of the vehicles till the preceding vehicle $(i-1,j)$ enters the intersection is kept unchanged. The order of the vehicles after the $(i-1,j)$ enters the intersection are changed. 
 
 The relative order of  any vehicle $(i^{\prime},j^{\prime})$ and  $(l,k)$ where $k\in \mathcal{M}_{j^{\prime}}$ is also kept unchanged if $(i,j)\neq (i^{\prime},j^{\prime})$. Thus, if the vehicle $(i^{\prime},j^{\prime})$ enters the intersection before the vehicle $(l,k)$ that order is maintained. Hence, one only needs to compute the optimal sequence among the newly entered vehicle $(i,j)$ and the vehicles at the conflicting lanes.   If there are $n$ number of vehicles in the conflicting lanes, we need to consider $n$-possible combinations in the worst case. Hence, the computational complexity reduces from exponential order to the  linear order of the number of vehicles in the conflicting lanes. We later show that such a relaxation also simplifies the solution of an optimization problem. 


The communication cost is also reduced. Note that the traffic controller needs to inform the AVs when they can enter the intersection if the current time needs to be updated because of the change in the schedule. Since potentially, at least the order of the vehicles till the vehicle $(i-1,j)$ enters the intersection is kept unchanged, thus, for those vehicles no additional information is required to be sent. 
\end{color}

Thus, we only compute $b_{i,j,l,k}$ with keeping other decisions $b_{i^{\prime},j^{\prime},l,k}$ intact. In the following, we describe the optimization problem which chooses the optimal $b_{i,j,l,k}$. 

\begin{color}{black}
\subsection{Objective}\label{sec:opti}
When the new vehicle $(i,j)$ arrives the traffic-controller decides $b_{i,j,l,k}$ only among $(i,j)$ and the vehicles $(l,k)$ where $l\in \mathcal{L}_k$, $k\in \mathcal{M}_j$, i.e., the vehicles in the conflicting lanes. 
We, first, introduce a notation which we use throughout this section. $\mathbf{b}_{-(i,j)}$ consists of all the decisions apart from the decisions $b_{i,j,l,k}$ for all $(l,k),\forall l\in \mathcal{L}_k,k\in \mathcal{M}_j$.

\subsubsection{Computation of the vehicles' dynamics}
As new vehicle $(i,j)$ arrives at time $t_{i,j}^{0}$ and the traffic-controller decides $b_{i,j,l,k}$ at time $t_{i,j}^0$. The traffic controller computes the profiles of the vehicles based on the decision $b_{i,j,l,k}$ since the other decisions are kept unchanged, $b_{i,j,l,k}$ is sufficient to compute the dynamics of the vehicles as described in Sections~\ref{sec:avlead} and \ref{sec:hdv}. Note that if $b_{i,j,l,k}=1$, the profile of the vehicle $(l,k)$ remains the same as it is computed before the vehicle $(i,j)$ enters the system. However, if $b_{i,j,l,k}=0$, the profile of the vehicle $(l,k)$ needs to be recomputed as the vehicle $(i,j)$ now enters the intersection before the vehicle $(l,k)$. The traffic-intersection controller only recomputes those velocity and positions of vehicles' which are changes because of $b_{i,j,l,k}$. 

 For vehicle $(l,k)$, the time when it can enter the intersection is given by (cf.(\ref{enter}))
\begin{align}\label{eq:entry_exit}
t_{l,k}^m(\mathbf{b})=\psi^m_{l,k}(t_{i,j}^{0},\mathbf{b})\nonumber\\
t_{l,k}^{f}(\mathbf{b})=\psi^f_{l,k}(t_{i,j}^{0},\mathbf{b}),
\end{align}
it indicates that the time is a function of the decision $b_{i,j,l,k}$. The position and velocity profiles are given by  (cf.(\ref{pos_vel}))
\begin{align}\label{eq:vel}
(v_{i,j}(t),p_{i,j}(t))=\zeta(t_{i,j}^{0},\mathbf{b})
\end{align}

If the vehicle $(i^{\prime},j^{\prime})$ has entered just before the the vehicle $(i,j)$ in the system. The  computed time is given by $t_{l,k}^{m}(\mathbf{b}_{-(i,j)})$ which does not include the decision $b_{i,j,l,k}$. Thus, because of the decision $b_{i,j,l,k}=0$, the additional time the vehicle $(l,k)$ now takes to enter the intersection is $t_{l,k}^m(\mathbf{b})-t_{l,k}^{m}(\mathbf{b}_{-(i,j)}$. 
\begin{definition}
We define the delay imposed by the newly entered vehicle $(i,j)$ on the vehicle $(l,k)$ as $\tau_{i,j,l,k}=t_{l,k}^m(\mathbf{b})-\tilde{t}_{l,k}^m(\mathbf{b}_{-(i,j)})$.
\end{definition}
$\tau_{i,j,l,k}=0$ if $b_{i,j,l,k}=1$ as the vehicle $(i,j)$ does not cause any change to the profile of the vehicle $(l,k)$. 

Note that even the vehicle in the non-conflicting lane $\tilde{j}$ of lane $j$ (e.g.,  lane $3$ of lane $1$) since the ordering of the vehicles remain intact between $(l,k)$ and $(\tilde{i},\tilde{j})$. Thus, similar to the vehicle $(l,k)$ we also define 
\begin{align}\label{eq:entry_exitnon-conflict}
t_{\tilde{i},\tilde{j}}^m=\psi^m_{\tilde{i},\tilde{j}}(t_{i,j}^{0},\mathbf{b})\nonumber\\
t_{l\tilde{i},\tilde{j}}^{f}=\psi^f_{\tilde{i},\tilde{j}}(t_{i,j}^{0},\mathbf{b})
\end{align}
Similarly, for the velocity and position profile
\begin{align}\label{eq:vel_nonconflict}
(v_{\tilde{i},\tilde{j}}(t),p_{\tilde{i},\tilde{j}}(t))=\zeta(t_{i,j}^{0},\mathbf{b})
\end{align}

\subsubsection{Energy Cost}
For any vehicle the energy cost is defined as the area under the curve of the magnitude of the change in velocity over the time the vehicle spends in the system, i.e., $\int_{t_{i,j}^{0}}^{t_{i,j}^f}u_{i,j}(t)^2dt$.
We denote the energy cost as $C_{l,k}(\mathbf{b})$ which explicitly depends on the decisions $\mathbf{b}$ since the vehicle dynamics, and $t_{i,j}^f$ are functions of $\mathbf{b}$. The cost inherently depends on the position and velocity profile of the vehicles. 

We will compute the energy cost because of newly arrived vehicle $(i,j)$. Note that all the vehicles $C_{m,j}(\mathbf{b})$ where $m<i$ since these vehicles' profiles are unchanged because of the newly entered vehicles. Similarly, if $b_{i,j,l,k}=1$, the profile of vehicle $(l,k)$ is unchanged because of the newly entered vehicle. This is the external cost incurred because of the newly entered vehicle $(i,j)$, also known as marginal cost a term which is popular in Economics. Even a vehicle $(i^{\prime},j^{\prime})$ which belongs to the non-conflicting lane of lane $j$, its dynamics can be changed because of the decision $b_{i,j,l,k}$. For example, if $b_{i,j,l,k}=0$, then the profile of the vehicle $(l,k)$ will be modified and if $b_{i^{\prime},j^{\prime},l,k}=1$ then, the velocity profile of the vehicle $(i^{\prime},j^{\prime})$ will be modified. 

We denote the change in the energy cost because of the newly entered vehicle and the decision $b_{i,j,l,k}$ as $C_{i^{\prime},j^{\prime}}(\mathbf{b})-C_{i^{\prime},j^{\prime}}(\mathbf{b}_{-{(i,j)}})$ where we denote $C_{i,j}(\mathbf{b}_{-{(i,j)}})$ as $0$ since the vehicle $(i,j)$ is the newly entered vehicle. Again, for the vehicles $(m,j)$ where $m<i$, this change in the cost is $0$ since we are not changing the order of those vehicles. Further, if $b_{i,j,l,k}=1$, then $C_{l,k}(\mathbf{b})-C_{l,k}(\mathbf{b}_{-{(i,j)}})=0$. 

\subsubsection{Formulation of the Optimization Problem}
We are now equipped with all the notations to describe the objective. We consider four objectives--
\begin{itemize}
\item Minimize the additional energy cost of the vehicles incurred by the vehicle $(i,j)$.
\begin{align}\label{eq:o1}
o_1=\sum_{j^{\prime}}\sum_{i^{\prime}\in \mathcal{M}_{j^{\prime}}}C_{i^{\prime},j^{\prime}}(\mathbf{b})-C_{i^{\prime},j^{\prime}}(\mathbf{b}_{-{(i,j)}})
\end{align}
\item Minimize the additional time taken by the vehicles to exit the system because of the vehicle $(i,j)$
\begin{align}\label{eq:o2}
o_2=\sum_{j^{\prime}}\sum_{i^{\prime}\in \mathcal{M}_{j^{\prime}}}t_{i^{\prime},j^{\prime}}^f(\mathbf{b})-t_{i,j}^f(\mathbf{b}_{-(i,j)})
\end{align}
\item The difference between the maximum speed and the  velocity of the vehicles at the intersection must be minimized
\begin{align}\label{eq:o3}
o_3=\sum_{j^{\prime}}\sum_{i^{\prime}}(v_m-v_{i^{\prime},j^{\prime}}(t_{i^{\prime},j^{\prime}}(\mathbf{b})))^{+}\nonumber\\-(v
_m-v_{i^{\prime},j^{\prime}}(t_{i^{\prime},j^{\prime}}(\mathbf{b}_{-(i,j)})))^{+}
\end{align}
\item The final time where all the vehicles would exit the intersection must be minimized. 
\begin{align}\label{eq:t_final}
o_4=t^{final}=\max_{i^{\prime},j^{\prime}}t_{i^{\prime},j^{\prime}}(\mathbf{b})-\max_{i^{\prime},j^{\prime}}t_{i^{\prime},j^{\prime}}^f(\mathbf{b}_{-(i,j)})
\end{align}
\end{itemize}
The first objective seeks to minimize the total energy cost incurred due to the latest scheduling change because of the new vehicle $(i,j)$. The second objective seeks to minimize time the vehicles take to exit the system. The third objective seeks to maintain the velocities of the vehicles close to the maximum speed limit when they enter the intersection. This will ensure that it maximizes the throughput. Final objective seeks to minimize the total time it takes for all the vehicles to exit the intersection. Since $\mathbf{b}_{-(i,j)}$ does not contain $b_{i,j,l,k}$, one can equivalently eliminate the terms which does not depend on $b_{i,j,l,k}$ and can represent everything as a function of $b_{i,j,l,k}$.

The optimization problem that the traffic controller solves is the following:
\begin{align}\label{eq:opti}
\mathcal{P}:\text{min }& o_1+\lambda_1o_2+\lambda_2o_3+\lambda_3o_4\nonumber\\
\text{subject to }& (\ref{eq:entry_exit}), (\ref{eq:vel}),  (\ref{eq:entry_exitnon-conflict}), (\ref{eq:vel_nonconflict})\nonumber\\
& b_{i,j,l,k}\in \{0,1\}
\end{align}
The decision variable is $b_{i,j,l,k}$. The objectives are defined in (\ref{eq:o1})-(\ref{eq:t_final}). Note that for some sequence, if there does not exist any feasible solution (as of Theorem 1), then, corresponding to the vehicle $(i^{\prime},j^{\prime})$, the $t_{i^{\prime},j^{\prime}}^{m}$ is $\infty$ and such a sequence will never be selected. Note that there always exists at least one feasible sequence where the vehicle $(i,j)$ will enter the intersection at the end. Hence, solution of the optimization problem is well-defined. 
\end{color}

\section{Traffic Controller's Decision Algorithm}
In this section, we discuss the algorithm for the traffic controller to determine the sequence when a new vehicle arrives in the control zone.  Subsequently, we demonstrate the algorithm using an illustrative example.

First, we introduce some notations.
\begin{definition}
Let $N_k$ be the number of vehicles at lane $k$.
\end{definition}
Hence, $N_k=|\mathcal{L}_k|$. 

Recall that $Seq(t)$ represents the order in which the vehicles present at time $t$ will enter the intersection. Let $S_{i,j}$ be the order in which the vehicles after the vehicle $(i,j)$ enters the intersection before the vehicle $(i+1,j)$ enters the intersection. Note that $Seq(t_{i+1,j}^{0}-)$ is the sequence of the vehicles just before the vehicle $(i+1,j)$ enters the intersection.  If $Seq(t_{i+1,j}^{0}-)(m)=(i,j)$  then, $S_{i,j}(m_1)=Seq(t_{i+1,j}^{0}-)(m+m_1)$.  If two  vehicles enter the intersection at same time from the non-conflicting lanes, then the tie is broken randomly.  Further, if the vehicle $(i,j)$ is the last one which enters the intersection, then $S_{i,j}$ is empty. $S_{i,j}(n)$ denotes the $n$-th component of the set $S_{i,j}$. Note that since the newly entered vehicle $(i+1,j)$ does not change the order of the vehicles till the vehicle $(i,j)$ enters the intersection, hence, in essence, we only need to compute the order of vehicles within $S_{i,j}$.

For example, in Fig~\ref{fig:zone}, if the sequence is $\{v_3,v_2,v_4\}$, then after the vehicle $v_3$ enters the intersection , the order for the rest of the vehicles is $\{v_2,v_4\}$.  

\subsection{Algorithm}\label{sec:algo}
\begin{enumerate}[leftmargin=*]
\item If the number of vehicle in the system is $i$ at lane $j$, when a new vehicle arrives in lane $j$, the new vehicle is represented as $(i+1,j)$. If multiple vehicles arrive at the same time, tie is broken arbitrarily. 
\item Compute  $t_{i+1,j}^{\prime}=t_{i+1,j}^{0}+\dfrac{L+S}{v_{free}}$. Initialize $n=1$. 
\item If $S_{i,j}(n)$ is empty, then skip the next steps and go to step 7.
\item If $S_{i,j}(n)=(i^{\prime},j^{\prime})$ for any  $i^{\prime}$ where $j^{\prime}\notin \mathcal{M}_j$, then we go to step 6. 
\item If $S_{i,j}(n)=(l,k)$, then we consider a possible decision $b_{i+1,j,m,k}[n]=1$ for all $m\leq l-1$ and $b_{i+1,j,m,k}[n]=0$ for $m\geq l$. Let $l^{\prime}_{k_1}=\arg\min\{l|S_{i,j}(n^{\prime})=(l,k_1), n^{\prime}>n\}$
 where $k_1\neq k, k_1\in \mathcal{M}_j$. Thus, the vehicle $(l^{\prime}_{k_1},k_1)$ at lane $k_1$ will wait till the vehicle $(i+1,j)$ enters the intersection under this schedule. Hence, $b_{i,j,m^{\prime},k_1}[n]=1$ for $m^{\prime}\leq l^{\prime}_{k_1}-1$ for $k_1\in \mathcal{M}_j, k_1\neq k$ and $b_{i+1,j,m^{\prime},k_1}=0$ for $m^{\prime}\geq l^{\prime}_{k_1}$. The order of the vehicles at lane $j^{\prime}$ where lane $j^{\prime}$ is non-conflicting to lane $j$ can be obtained since the order $b_{i^{\prime},j^{\prime},l,k}$ $k\in \mathcal{M}_j^{\prime}$ remains the same as it was before the vehicle $(i,j)$ enters the control zone. Compute  $t_{i+1,j}^{m,temp}[n]=t_{i+1,j}^{0}+\dfrac{L}{v_{free}}, t_{i+1,j}^{f,temp}[n]=t_{i+1,j}^{\prime}$. The velocity profile $v_{i+1,j}^{temp}[n](t)$ and position profile $p_{i+1,j}^{temp}[n](t)$ are also computed . All the other vehicles' velocity and position profiles are also computed and are stored in $t_{i^{\prime},j^{\prime}}^{m,temp}[n]. t_{i^{\prime},j^{\prime}}^{f,temp}[n],v_{i^{\prime},j^{\prime}}^{temp}[n]$, and $p_{i^{\prime},j^{\prime}}^{temp}[n](t)$. Compute the objective value. 
 \item If $S_{i,j}(n)=(i^{\prime},j^{\prime})$ where $j^{\prime}\notin \mathcal{M}_j$, then $obj[n]=\infty$, otherwise  store the objective value in step 5 at  $obj[n]$. Make $n=n+1$. 
\item Let $z=\arg\min _n obj[n]$. Update $t_{l,k}^{m}=t_{l,k}^{m,temp}[z], t_{l,k}^{f}=t_{l,k}^{f,temp}[z], v_{l,k}(t)=v_{l,k}^{temp}[z](t), p_{l,k}(t)=p_{l,k}^{temp}[z](t)$ $\forall l\in \mathcal{N}_{j,k}\forall k\in \mathcal{M}_j$, $t_{i+1,j}^{m}=t_{i+1,j}^{m,temp}[z], t_{i+1,j}^{f}=t_{i+1,j}^{f,temp}[z]$, and $v_{i+1,j}(t)=v_{i+1,j}^{temp}[z](t)$, $p_{i+1,j}(t)=p_{i+1,j}^{temp}[z](t)$. The sequence order is updated as $b_{i+1,j,l,k}=b_{i+1,j,l,k}[z]$. 
\item The trajectory and dynamics of vehicles are updated. Update the sequence $Seq(t_{i+1,j}^0)$. If the vehicle $(i^{\prime},j^{\prime})$'s order is changed from $Seq(t_{i+1,j}^{0}-)$,  $j^{\prime}_{i^{\prime}}$ is different from $j^{\prime}_{i^{\prime}-1}$ and  $(i^{\prime},j^{\prime})$ is an AV, then the traffic-controller informs the updated time the AV can now enter the intersection.
\end{enumerate}
We obtain--
\begin{theorem}
The algorithm gives a solution of the optimization problem $\mathcal{P}$ in (\ref{eq:opti}). 
\end{theorem}
\begin{proof}
Since the algorithm picks the sequence among all the possible combinations for which the objective is the smallest, thus, it is a solution of Optimization problem in (29)
\end{proof}

\subsection{Discussion}
When a vehicle enters the system, the traffic controller computes the current trajectory of the vehicle using the vehicle dynamics as discussed in Section~\ref{sec:system_model}. The algorithm tries to find the optimal $b_{i+1,j,l,k}$ for $l\in \mathcal{L}_k$, $k\in \mathcal{M}_j$. The relative order of the other vehicles is kept the same. The algorithm enumerates over all the possible sequence among the newly entered vehicle and the vehicles at the conflicting lanes in the original computed sequence after the vehicle $(i+1,j)$ enters the intersection. Hence, the traffic-controller only needs to find the optimal position of the vehicle $(i+1,j)$ among the vehicles in $S_{i,j}$. 

We select the first item of $S_{i,j}$. Suppose that it is $(l,k)$. Now, we compute the value of the objective function when the vehicle  $(l,k)$ needs to wait as the vehicle $(i,j)$ enters the intersection before vehicle $(l,k)$, i.e, $b_{i,j,l,k}[1]=0$. The rest of the order of $S_{i,j}$ remains the same.  Hence, if this is the decision made by the traffic-controller, then the sequence is $S^{new}(m)=S^{old}(m)$ for $m\leq n$ where $S^{old}(n)=(i,j)$, $S^{new}(n+1)=(i+1,j)$ and $S^{new}(m)=S^{old}(m-1)$ for $m>n+1$. The objective value, and the vehicle dynamics are  stored. Then, we pick the second element of $S_{i,j}$. Let the second element be $(l,k_1)$. Now we compute the objective value when $b_{i,j,l,k_1}[2]=0$. However, $b_{i,j,l,k}[2]=1$. If this is the optimal sequence chosen by the traffic-controller, then $S^{new}(m)=S^{old}(m)$ for $m\leq n+1$ where $S^{old}(n+1)=(l,k)$. $S^{new}(n+2)=(i,j)$ and $S^{new}(m)=S^{old}(m)$ for $m>n+2$.  We again store the objective values, and the dynamics. If $S_{i,j}(n)=(i^{\prime},j^{\prime})$ for any $i^{\prime}$, we do not take any action as we are only interested in scheduling of vehicles in the conflicting lanes. We continue this process until we exhaust  $S_{i,j}$. We then select the decision which corresponds to the lowest objective value.

We now illustrate the algorithm using an example. 
\begin{example}
Consider Fig.~\ref{fig:zone}. The vehicle $v_1$ has left the intersection. Let us assume that the current schedule of entering the intersection ($S^{old}$) be $\{v_3,v_2,v_4\}$ which means that $v_3$ would enter the intersection first, then $v_2$, and after that $v_4$ would enter the intersection. The corresponding estimated time to enter the intersection and leaving the intersection is computed by the traffic controller from the vehicle dynamics. Now, suppose vehicle $v_5$ enters the intersection at time $t_5^{0}$. Since the vehicle $v_5$ is the only vehicle in the lane, it would move at the speed limit. %

$v_3$ belongs to the lane which is non-conflicting to lane $3$. Thus, we do not select the order between $v_3$ and $v_5$ as they can enter the intersection without colliding.First according to our notation $v_5=(1,3)$, $v_2=(1,2)$, $v_4=(1,4)$, and $v_3=(1,1)$. We, thus, first consider the decision $b_{1,3,1,2}=0$, i.e, the vehicle $v_5$ will enter the intersection before the vehicle $v_2$ since $v_2$ enters the intersection before $v_4$, thus, $b_{1,3,1,4}=0$. The traffic-controller computes the objective for this decision. Now, we consider the decision $b_{1,3,1,4}=0$, however, $b_{1,3,1,2}=1$. We compute the value of the objective. Finally, we consider the decision $b_{1,3,1,4}=1$, $b_{1,3,1,2}=1$. Hence, $v_5$ will enter the intersection at the last position if this decision is implemented. We compute the value of the objective by computing the dynamics of the vehicles.

We see which combination gives the smallest value. That combination is implemented and the optimal sequence is updated. For example, if the decisions $b_{1,3,1,2}=1$ and $b_{1,3,1,4}=0$ give the smallest objective value. Then, the optimal sequence is $\{v_3,v_2,v_5,v_4\}$.

\end{example}
\begin{color}{black}
\section{Generalizations}
\subsection{Response time of the HDVs and Amber-light}\label{sec:amber}
There may be a non-zero reaction time of the HDVs when they see change in the traffic-light. We can easily accommodate the reaction time when the phase changes from the red to green  by modifying our model in the following manner. When the HDV $(i,j)$ faces green-light after the red-light, the HDV will respond after a reaction time of $t_{react}$. Thus, if the green-light is switched on at time $t^g$, then for time till $t^g+t_{react}$, the HDV will respond as if it is still facing the red-light, after the time $t^g+t_{react}$, the HDV will then follow the dynamics when it faces the green-light as described in (\ref{eq:dynh_acc}).  Hence, similar algorithm as described can be used to obtain the optimal sequence.

When the green-light is changed to red-light at a lane, in general, an amber-light is used before the transition. We can accommodate an amber-light duration in the following manner for the HDV. When the vehicle $(i,j)$ is at most $\delta$-distance away from entering the intersection, we assume that the vehicle $(i,j)$ will follow its dynamics as if it is facing a green-light and enter the intersection, i.e., it will treat the amber-light as the green-light. On the other hand when the vehicle $(i,j)$ is at least $\delta$-distance away from the intersection, the vehicle $(i,j)$ will not enter the intersection and stop similar to the scenario where it is facing the red-light. The distance $\delta$ is chosen as the distance that can be traversed before the vehicle $(i,j)$ can react and then applies the maximum deceleration possible (cf.(\ref{eq:lowerumin})) and still will be able to stop. Since when the HDV $(i,j)$ enters the intersection when it is at most $\delta$-distance away from the intersection when the amber-light is switched on, we consider that if $j_i(k)$ and $j_{i-1}(k)$ are not the same for some $k\in \mathcal{M}_j$, then the amber-light itself is switched on when the vehicle $(i-1,j)$ is $\delta$-distance away from the intersection as it would enforce that the vehicle $(i,j)$ will treat the amber-light as the red-light and the amber-light will transition to the red-light without affecting the schedule.\footnote{Still it may happen that a vehicle would enter the intersection if it is more than $\delta$-distance away, the traffic-intersection controller can sense it and put a red-light on all the intersection to allow the vehicle to cross the intersection. The traffic-intersection controller would remove the vehicle from the schedule and would update the schedule by retaining the relative order of other vehicles intact.}

\subsection{Deviation from the dynamics}\label{sec:deviation}
The HDVs'  parameters  may not be the same and may not be known to the traffic-controller. We now discuss how our model can accommodate the setting where the parameters of the HDVs' are drawn from a known and bounded distribution. In this case, the optimal sequence is obtained based on the expected values of the parameters. Note that once the sequence is decided, the traffic-controller also decides when to switch on the red or green-light at a lane based on the decision. However, since the parameter values may not coincide with the estimated values, the time at which the red-light and green-light are switched on will depend on the realized values and the decision $b_{i,j,l,k}$. 

Note that the traffic-controller provides the exact time when the AV can enter the intersection in the case it becomes the lead-vehicle and faces the red-light. When the parameters of the HDVs are drawn from a distribution, our model can be extended by considering the worst-case scenario. In particular, when the parameters of the HDVs follow a distribution, the time at which a HDV exits an intersection then can be represented by a probability distribution function for a given decision $\mathbf{b}$. The worst-case scenario is then the maximum value $\overline{t}_{l,k}^f$ of the support set of that distribution for a vehicle $(l,k)$. The traffic-intersection controller then computes $\overline{t}^g_{i,j}$ if $j_i(k)$ and $j_{i-1}(k)$ are different where $\overline{t}^g_{i,j}$ is computed similar to $t^g_{i,j}$ with $\overline{t}_{l,k}^f$ in place of $t_{l,k}^f$. $\overline{t}_{l,k}^f$ is the minimum time by which the vehicle $(l,k)$ would exit the intersection w.p. $1$. Hence, $\overline{t}^g_{i,j}$ represents the minimum time at which the vehicle $(j_i(k),k)$ and $(j_i(k_1),k_1)$ where $k,k_1\in \mathcal{M}_j$ would exit the intersection with probability $1$.

The traffic-intersection controller then informs the time $\overline{t}^g_{i,j}$ to the vehicle $(i,j)$  can enter the intersection which will ensure that AV will not collide with the vehicles at the intersection. Note that as the traffic-controller samples the vehicles' positions and velocities, those values can be updated. This approach is similar to {\em robust optimization approach} in Stochastic Optimization. Detailed numerical results are presented in Section~\label{sec:deviate}. Note that in the IDM model, it is easy to generate the worst-case scenario. For example,  the maximum possible value of $T_{i,j}$ (the head-way distance between two vehicles) will correspond to the worst-case scenario. 

Note that in the similar way, our model can accommodate the stopping of HDVs. One can again we can formulate a distribution of the time at which a HDV will enter and exit the intersection based on the stopping behavior of the HDVs. However, a complete characterization has been left for the future. 

Further, our model can also be extended to the scenario where the AV's dynamics are also different from the ones the traffic-controller estimates. As we discussed, the traffic-controller would decide the optimal sequence based on the estimated values of the parameters. The traffic-controller can also estimate the parameter for an AV depending on its behavior. Developing the complete framework is left for the future. 

In the rare scenarios where a vehicle may stop permanently because of a break-down, the vehicle is removed from the sequence order and the order of the rest of the vehicles are maintained as it is. When again the vehicle is again up for running, the vehicle order is updated by treating the vehicle as a new vehicle. 

\subsection{Variation of Maximum Speed}\label{sec:free_velocity}
Though we have considered in (\ref{eq:min}) that AV must enter the intersection at the maximum speed when it solves the optimal control problem. We can relax the assumption. For example, if there is a congestion at the neighboring intersection, the traffic controller may want to reduce the speed of the vehicles entering the intersection, then $v_m$ may not be the maximum speed, rather $v_m$ may be the desirable speed decided by the traffic-intersection controller. Note that the vehicles entering the intersection can also enter the intersection at a speed below the maximum speed because of the congestion. We, empirically, evaluate the impact of $v_m$ on the results in Section~\ref{sec:diff_speed}.

\subsection{Allowing Turning at the intersection} \label{sec:turn}
Our model can be extended to accommodate the scenario where vehicles can turn. We need to partition the set of lanes from where the vehicles can enter the intersection simultaneously, i.e., the set of non-conflicting lanes (for example, lanes $1$ and $3$ in Fig.~\ref{fig:zone}) depending on the nature of the lanes. We can use the conflict graph notion as described in \cite{hajbabaie} where the nodes are the lanes and edge exists between two nodes if they are conflicting to each other. A maximal independent set in a conflict graph represents a set of non-conflicting lanes. 

The safety constraint will be modified where the vehicles from the conflicting lanes can not enter the intersection while one vehicle is already in the intersection. This also characterizes the decision $b_{i,j,l,k}$ which schedules the vehicles from the conflicting lanes. We can characterize the dynamics of turning vehicles using the model proposed in \cite{hajbabaie}. Subsequently, the traffic controller can decide the optimal sequence among the vehicles in the conflicting lanes. The detailed characterization has been left for the future. 

\subsection{Difference between the optimal solution and solution of the relaxed problem}
In this section, we illustrate difference between the optimal schedule and the optimal schedule obtained by solving the relaxed problem as described in $\mathcal{P}$ (cf.(\ref{eq:opti})) using an example. Since in order to determine the optimal schedule, one needs to consider every possible combination, we consider a set-up with only $3$ vehicles.

Consider a two lane intersection which intersect with each other. In Fig.~\ref{fig:zone} consider the two lanes only lanes $1$ and $2$ there. $3$ vehicles arrive at lane $1$ at times $0,2$ and $7$-th seconds respectively. $3$ vehicles arrive at lanes $2$ at times $1$, $4$, and $15$-th seconds. The length $L=300$m and $S=100$m. Further $v_m=20$m/s. Hence, it would take $20$ seconds to cross the entire intersection at the maximum speed without any traffic.

We also assume that the third vehicle at lane $1$ (i.e., the vehicle which arrives at the $7$-th second) is AV and the first vehicle at lane $2$ (the vehicle arrives $1$-nd second) is the AV. The rest are HDV. We set $\lambda_1=0$, $\lambda_2=0.01$, $\lambda_3=0$, and $\lambda_4=10$. Thus, the scheduler puts more weight on the objective where the time the last vehicle would exit the intersection. The weights on the first and third objectives are $0$. 

The schedule attained by our proposed algorithm is $\{v_{1,1},v_{1,2},v_{2,1},v_{2,2},v_{2,3},v_{1,3}\}$. In our proposed algorithm, first the two vehicles arriving at the lane $1$ would enter the intersection even the second vehicle at lane $1$ enters after the first vehicle at lane $2$. Since the third vehicle at lane $1$ is an AV, thus, it can enter the intersection at the maximum speed while the vehicles at the lane $2$ can be allowed to enter the intersection even the third vehicle at lane $2$ arrives $8$-seconds later than the third vehicle at lane $1$. The final time where the last vehicle exits the intersection is at $40$-th seconds. 

On the other hand the optimal schedule is $\{v_{1,1},v_{1,2},v_{1,3},v_{2,1},v_{2,2},v_{2,3}\}$. Thus, first, the vehicles at lane $1$ enter the intersection and then the vehicles at lane $2$ enter the intersection. Note that the third vehicle at lane $2$ is a HDV and enters $8$-seconds after $v_{1,3}$. Hence, in the schedule proposed by our algorithm, the vehicle $v_{1,3}$ has to wait for the vehicle $v_{2,3}$ to exit the intersection. The final time at which the last vehicle in this sequence exits the intersection is at $36$-th second which is $4$-seconds faster compared to the one obtained by our algorithm. Note that in our proposed algorithm, only the scheduling involving the newly entered vehicle and the vehicle at the conflicting lanes are considered without changing the relative order of other vehicles. Thus, when the vehicle $v_{2,3}$ arrives then the optimal schedule between $v_{1,3}$ and $v_{2,3}$ are considered. Specifically, the schedule $\{v_{1,1},v_{1,2},v_{2,1},v_{2,2}\}$ were kept unchanged as the schedule till the time $v_{2,2}$ enters the intersection is kept unchanged. The two possible schedules are considered--i) $\{v_{1,1},v_{1,2},v_{2,1},v_{2,2},v_{2,3},v_{1,3}\}$ (corresponding to $b_{2,3,1,3}=0$) and ii) $\{v_{1,1},v_{1,2},v_{2,1},v_{2,2},v_{1,3},v_{2,3}\}$ (corresponding to $b_{2,3,1,3}=1$). The first combination yields the final time as $40$ seconds whereas the second combination yields the final time as $44$ seconds. Thus, our proposed algorithm chooses the combination i). However, in the optimal scheduling algorithm all the possible combinations are considered every time a new vehicle arrives. Thus, the schedule is different. However, the increase in final time is not significant yet out proposed algorithm is computationally easy. 

\end{color}
\section{Numerical Results}\label{sec:numerical}
\subsection{Set Up}

We consider an intersection with four lanes as depicted in Fig.~\ref{fig:zone}. We assume that $L=300m$, $S=100m$, $v_{min}=0$,and $v_{m}=20m/s$. $a_{max}=u_{max}=10 m/s^2$ and $a_{min}=u_{min}=-10 m/s^2$ are respectively the maximum and minimum acceleration of any type of vehicle. We also consider $a_h=2m/s^2$. Recall that $a_h$ is the acceleration when the lead vehicle is HDV and enters the intersection at a speed below the speed limit. The velocity of any vehicle can not exceed $v_m$. The distance $S$ is enough to attain the highest speed from the zero velocity if the vehicle accelerates uniformly with $a_h$.

When the vehicles follow another vehicle, we also need to set the parameters for the IDM. We assume that $T_{i,j}=2.5$-sec for the HDVs, and $T_{i,j}=1.5$-s for the AVs. We investigate the impact when the $T_{i,j}$ of a HDV varies randomly following a distribution in Section~\ref{sec:deviate}.  $s_0$ the safe steady distance between any two vehicles is assumed to be $2$-m. 

For the simulation, we discretize the velocity--
\begin{align}\label{vehadiscrete}
v_{i,j}(t+\Delta t)-v_{i,j}(t)= \dfrac{dv_{i,j}(t)}{dt}\Delta t \nonumber\\
p_{i,j}(t+\Delta t)-p_{i,j}(t)= v_{i,j}(t)\Delta t+\nonumber\\
 \dfrac{1}{2}(v_{i,j}(t+\Delta t)-v_{i,j}(t))^2\Delta t
\end{align}
where $\Delta t$ is the small time for discretization. We set $\Delta t=0.01$-sec. We also set $l=100m$ which is the largest distance from which a HDV would start decelerating when it encounters a red light. Note that the value of $\dfrac{dv_{i,j}(t)}{dt}$ is based on whether the vehicle is following another vehicle (IDM model, cf.(\ref{eq:IDM_follower_green})), or, it is the lead vehicle as in ((\ref{eq:lowerumin}), and (\ref{leadh})). For the AV, the value is either given by (\ref{eq:leadavdyn}) or by (\ref{eq:IDM_follower_green}). \begin{color}{black} We also consider that when the traffic-light transitions from red to green a HDV responds after 1-sec. \end{color} 

We consider that the vehicles arrive according to a Poisson process in each lane. The vehicles arrive with mean inter-arrival time of $10$ seconds. We run the simulation for 1 hour i.e., 3600 seconds. For each vehicle, with probability (w.p.) $p$  we assign it as a HDV and  w.p. $1-p$ we assign it as a AV. We compare our proposed algorithm with the FIFO algorithm where the vehicles which have entered first would cross the intersection first. We take average over $100$ simulations. \begin{color}{black}We compare our algorithm with respect to the FIFO algorithm\cite{dresner,cassandras1}. FIFO algorithm is the following-- it ranks the vehicles according the order they arrive at the control zone, thus, the vehicle who enters the control zone first is considered to enter the intersection first. The tie is broken in an arbitrary manner. \end{color}

\subsection{Results}
\subsubsection{Total time}
Fig.~\ref{fig:time} shows the variation of the  time at which the last vehicle would cross the intersection (i.e.$o_4$, cf.(\ref{eq:t_final})). This would measure the total time taken by all the vehicles to cross the intersection. Note that as the percentage of the AVs increases, $t^{final}$ decreases. Intuitively, as the percentage of the AVs increases the traffic controller can send signals when they can cross the intersection. Since the AVs can optimize the velocity profile accordingly, the total time taken by all the vehicles also reduce. 

Fig.~\ref{fig:time} also depicts that compared to the FIFO, our algorithm decreases the time significantly. The time reduction is higher when the percentage of the AV vehicles is small. Since in the FIFO, if the number of HDVs is large, the HDVs need to stop and wait for the vehicles in the conflicting lanes to cross before it can cross. The HDV then needs to accelerate and crosses the intersection at a slower speed compared to the AV. Hence, the total time taken in the FIFO becomes very large when the percentage of the AVs is small. Thus, our algorithm is most effective when the penetration rate of the AVs is small and our algorithm provides improvement of over 90\%. Even when the percentage of the AVs is 90\%, we obtain 50\% decrease in the total time taken by the vehicles to cross the intersection. 

Recall from the optimization problem in (\ref{eq:opti}) that $\lambda_3$ corresponds to the objective of the total time. Hence, as $\lambda_3$ increases compared to $\lambda_1$ and $\lambda_2$, our algorithm provides lower total time taken by all the vehicles to cross the intersection. The total number of vehicles which cross the intersection is significantly higher compared to the FIFO. Hence, our algorithm provides significantly better throughput compared to the FIFO.

\begin{figure}
\begin{center}
\includegraphics[width=0.3\textwidth]{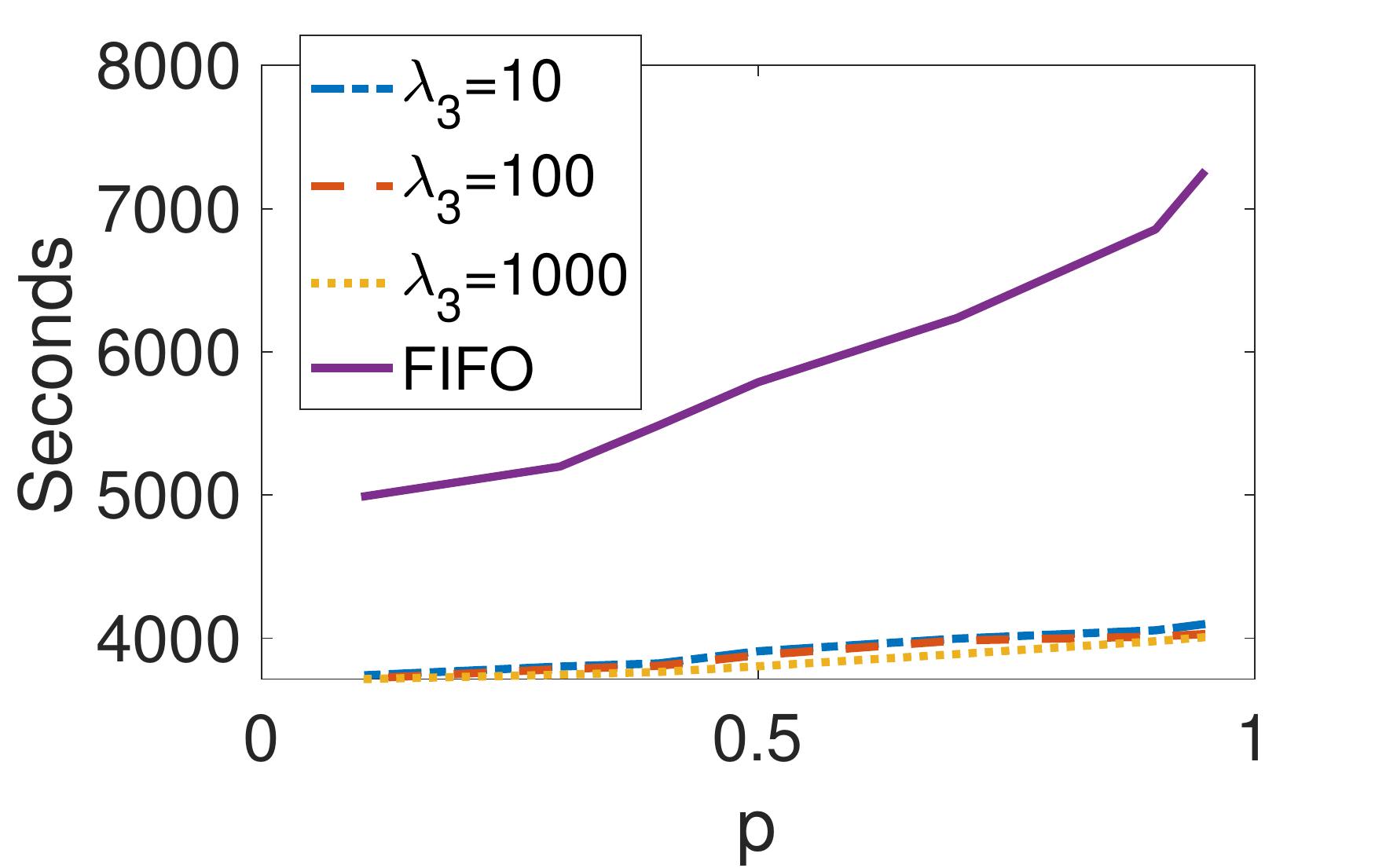}
\vspace{-0.1in}
\caption{The variation of time at which the last vehicle would cross the intersection with the AV penetration rate (i.e, $1-p$). $\lambda_1=100,\lambda_2=100$. }
\label{fig:time}
\end{center}
\vspace{-0.1in}
\end{figure}
\subsubsection{Distribution of the Delay}
Fig.~\ref{fig:delay} shows the average delay over all vehicles. Recall that the delay is 
$t_{i,j}^{f}-t_{i,j}^{\prime}$ (Definition~\ref{def:delay}).  As the percentage of the AVs increases, the delay decreases, since if the AVs are sent signals regarding the time they would cross the intersection, they would optimize the dynamics in order to enter the intersection at the speed limit which would reduce the delay for all the vehicles. 

\begin{figure}
\begin{center}
\vspace{-0.2in}
\includegraphics[width=0.3\textwidth]{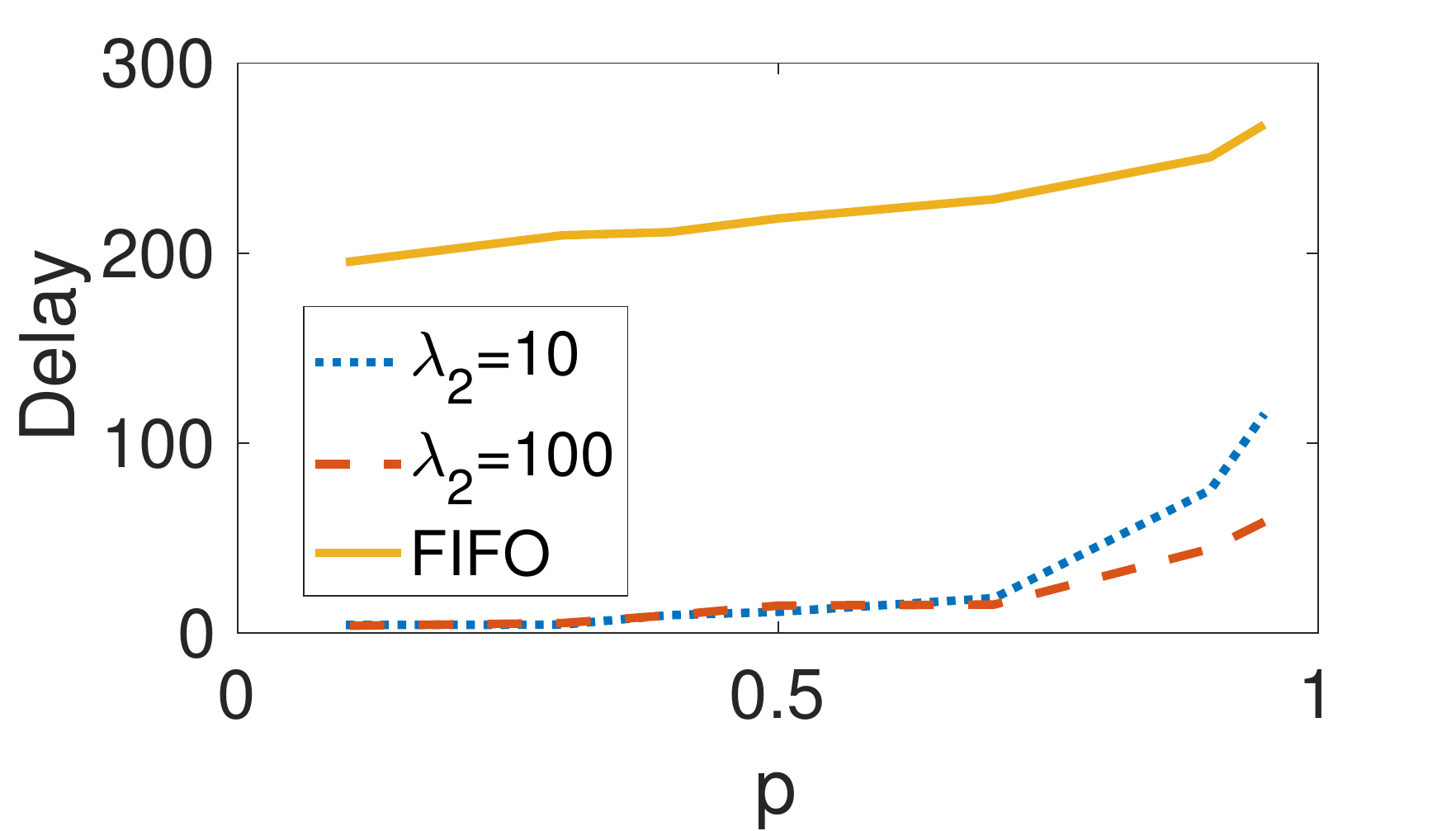}
\caption{The variation of the average delay over vehicles with the AV penetration rate (i.e, $1-p$). $\lambda_1=100,\lambda_2=100$. }
\label{fig:delay}
\end{center}
\vspace{-0.2in}
\end{figure}
Fig.~\ref{fig:delay} shows that compared to the FIFO, our algorithm gives significant improvement. For example, when the AV penetration rate is only 5\%,the average delay is around 500 seconds for FIFO whereas in our algorithm, it is only around 120 seconds for $\lambda_2=10$ and around 90 seconds when $\lambda_2=100$. Even when the AV penetration rate is 90\% the average delay in FIFO is around 195 seconds, where as in our algorithm is merely 5 seconds (close to negligible). Thus, our algorithm reduces the mean delay significantly compared to the FIFO system. 

Recall from the optimization problem in (\ref{eq:opti}) that $\lambda_2$ corresponds to the penalty for larger delay. Hence, as $\lambda_2$ increases compared to $\lambda_1$ and $\lambda_3$, the mean delay also decreases (Fig.~\ref{fig:delay}). 

Fig.~\ref{fig:delay2} shows that compared to the FIFO, the maximum delay among all the vehicles is significantly less in our algorithm. Maximum delay or the worst case delay measures the robustness of our algorithm. Note that as the percentage of the AVs increases the maximum delay  decreases since the AVs can optimize the dynamics and can follow closely compared to the HDVs.  On the other hand, the HDVs stop at the red signal and then again starts when the signal is green which results in greater delay. Fig.~\ref{fig:delay2} depicts that the worst case delay can be in the order of $3600$ seconds (i.e., $1$ hour) when the AV penetration rate is only 5\%. However, still in our algorithm it is of the order of 540 seconds ($9$ minutes). Fig.~\ref{fig:delay2} reveals that even when the AV penetration rate is 90\%, the worst case delay can be of the order of $1200$ seconds ($20$ minutes), however, in our algorithm, it is of the order of $100$ seconds (i.e., $2$ minutes). Thus, our algorithm significantly reduces the worst case delay compared to the FIFO. 

Similar to Fig.~\ref{fig:delay}, in Fig.~\ref{fig:delay2}, as $\lambda_2$ increases compared to $\lambda_1$ and $\lambda_3$, the worst case delay decreases. However, the decrement is not very significant.

\begin{figure}
\begin{center}
\vspace{-0.1in}
\includegraphics[width=0.3\textwidth]{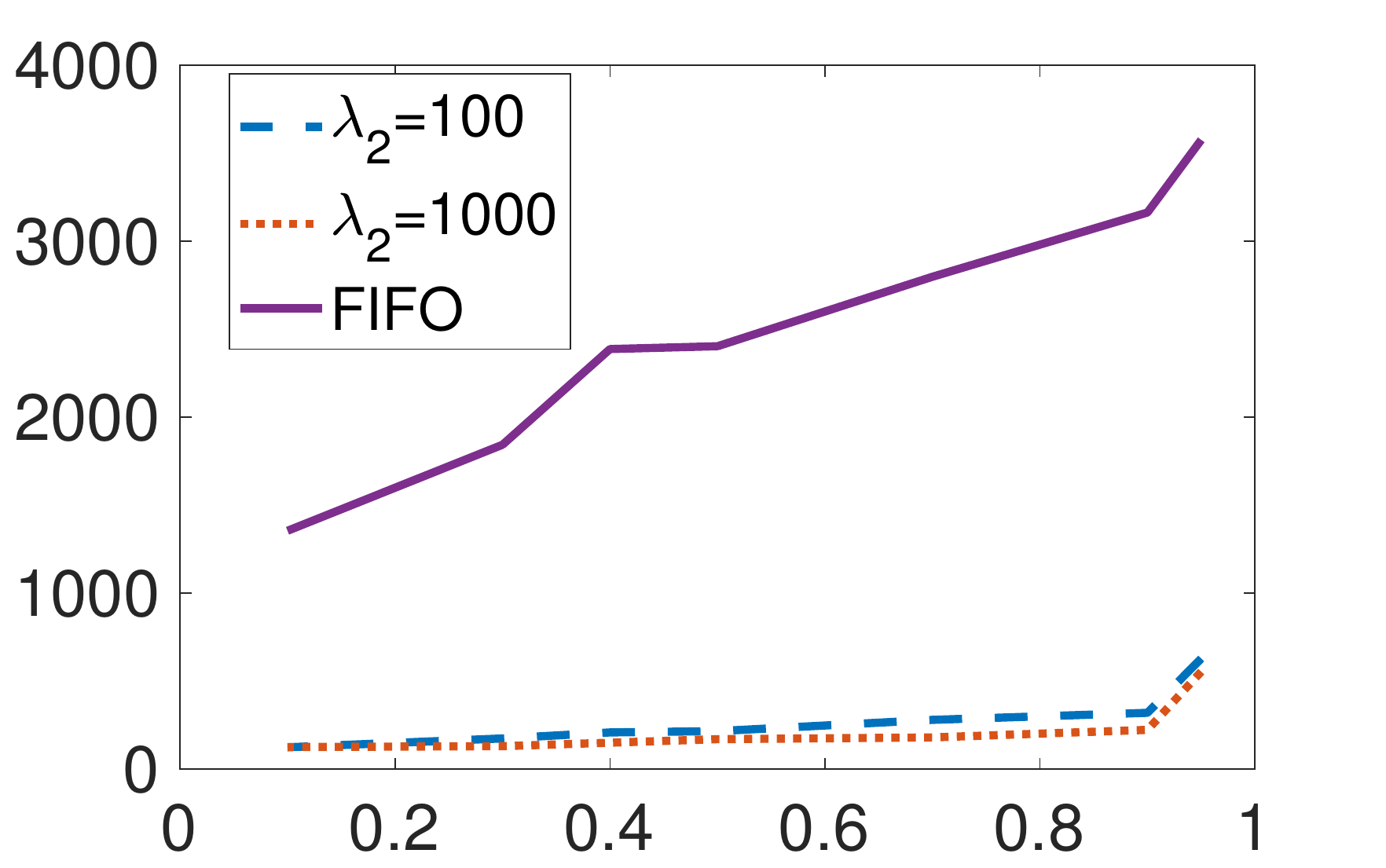}
\caption{The variation of the maximum delay of vehicles with the AV penetration rate (i.e, $1-p$). $\lambda_1=100,\lambda_2=100$. }
\label{fig:delay2}
\end{center}
\vspace{-0.25in}
\end{figure}

\subsubsection{Distribution of the Velocity while entering the Intersection}
Fig.~\ref{fig:vel2} depicts the average  velocities over the vehicles when they enter the intersection. Recall that the velocity at which the vehicle $i$ enters the intersection at lane $j$ is given by $v_{i,j}(t_{i,j}^{m})$. If the traffic controller sends information to the AVs, they can cross the intersection at the maximum speed. Hence, as the percentage of the AVs increases more vehicles cross the intersection at the full speed. Note that the vehicles following the lead vehicle try to maintain the same speed as that of the lead vehicles. Hence, if the lead vehicle crosses at the fullest speed, the lead vehicles also tend to do that which results in higher percentage of vehicles entering the intersection at the highest speed. Thus, as the AV penetration rate increases, the average velocity at which the vehicles enter the intersection increases. 

Fig.~\ref{fig:vel2} shows that the our algorithm ensures better mean velocity compared to the FIFO where the average is taken over all the vehicles. The performance of our algorithm is significantly better when the AV penetration rate is small. For example, when the AV penetration is only 5\%, the mean velocity in the FIFO is only around $1.5m/s$. However, our algorithm gives mean velocity around $5.5m/s$. As the percentage of the AVs increases, the mean velocity increases in FIFO as well,  still the mean velocity is significantly smaller compared to our algorithm. Recall from the optimization problem in (\ref{eq:opti}),  $\lambda_1$ corresponds to the weight regarding the penalty for velocity is not equal to the free speed limit. Hence, as $\lambda_1$ increases compared to $\lambda_2$ and $\lambda_3$, the mean velocity increases. 
\begin{figure}
\begin{center}
\includegraphics[width=0.3\textwidth]{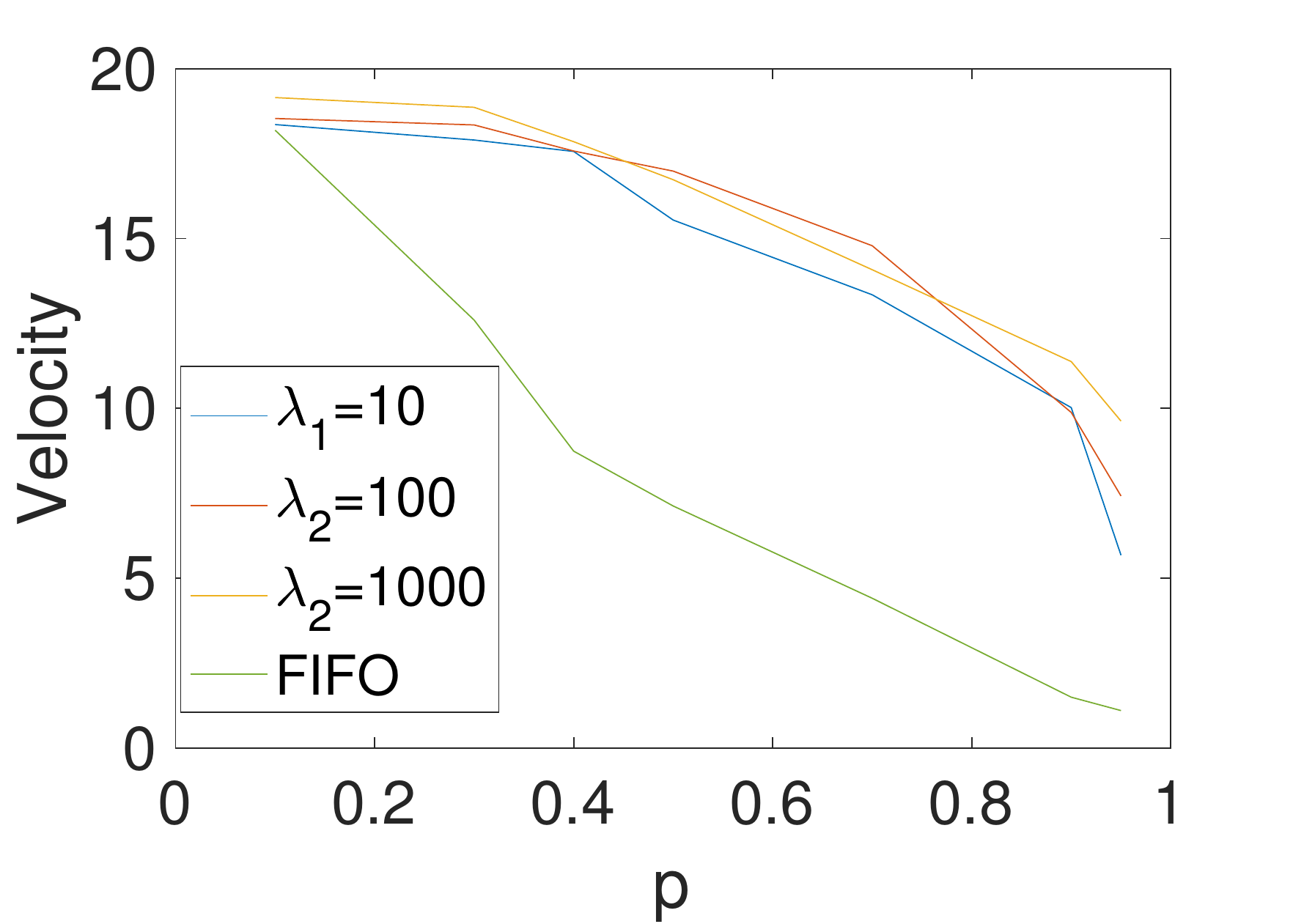}
\vspace{-0.1in}
\caption{The variation of the velocity of vehicles at which they enter the intersection zone with the AV penetration rate (i.e, $1-p$). $\lambda_1=100,\lambda_2=100$. }
\label{fig:vel2}
\end{center}
\vspace{-0.3in}
\end{figure}

\subsubsection{Acceleration and Deceleration}
Fig.~\ref{fig:acc} depicts that the probability mass functions of the absolute value of acceleration (or, retardation) over all the vehicles during the one hour window. Our result indicates that as the penetration of the AV increases (i.e. $p$ decreases), the vehicles tend to accelerate or retard less. Hence, the acceleration is more concentrated towards the value $0$ when the penetration of the AVs increases. Intuitively, the AVs would optimize its path when it is sent information pertaining when it can enter the intersection. Further, the AVs can adapt to the velocity at which the preceding vehicle is moving in a better way compared to the HDVs.  

\begin{figure}
\begin{center}
\includegraphics[width=0.3\textwidth]{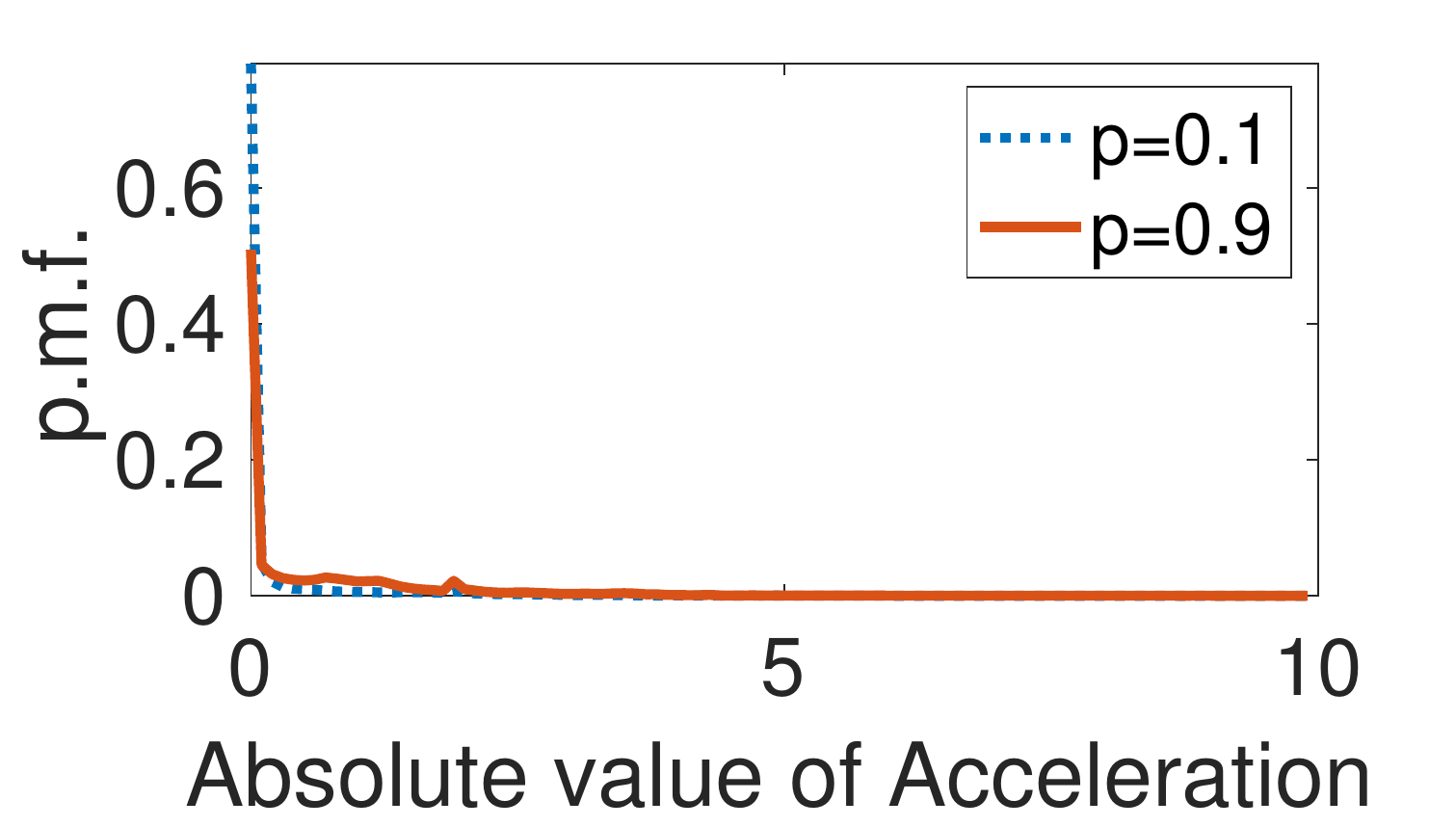}
\vspace{-0.1in}
\caption{The probability mass function of absolute value of acceleration or retardation as a function of $p$. }
\label{fig:acc}
\end{center}
\vspace{-0.3in}
\end{figure}

\subsubsection{Sensitivity Analysis}\label{sec:sensitivity}
Throughout the paper, we consider that the traffic controller predicts the dynamics of the vehicles when a new vehicle enters. In this section, we investigate the scenario where there is an error in prediction regarding the positions of the vehicles. Specifically, we add a random noise with the predicted position of the vehicles. The noise is halved for the AVs. We assume that the random noise is uniformly distributed within the range $[-r,r]$. The noise is also tuned in order to maintain a safe distance with the preceding vehicles. 

Fig.~\ref{fig:sensitive} indicates that even when the traffic controller has a noisy prediction, the performance of the algorithm does not vary much. For example, when the location uncertainty is around $40$m, the total travel time only increases by 5 seconds in 3600 seconds window. When the AV penetration rate increases, the impact is even less since the uncertainty pertaining the positions of the AVs are less compared to the HDVs. Further, as the uncertainty increases, the performance degrades, however, the degradation is negligible. Thus, our proposed algorithm can adapt effectively to the uncertainties in estimating the position of the vehicles. 

\begin{figure}
\begin{center}
\includegraphics[width=0.3\textwidth]{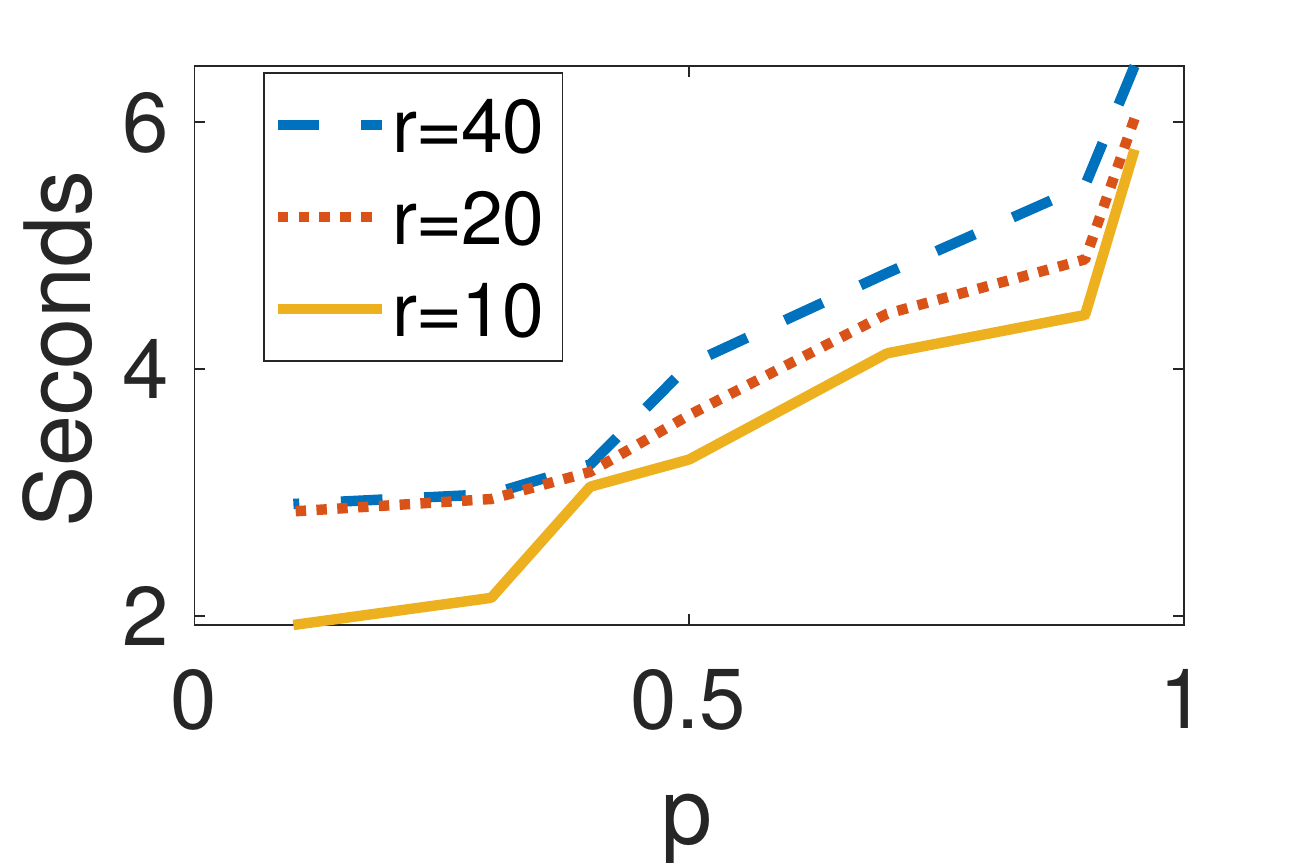}
\vspace{-0.1in}
\caption{The variation of the increase in total travel time as a function of $r$ and $p$.}
\label{fig:sensitive}
\end{center}
\vspace{-0.3in}
\end{figure}

\begin{color}{black}
\subsection{Deviation from the nominal value}\label{sec:deviate}
Note that so far, we assume that the traffic-intersection controller knows the parameter values of the HDVs. However, those values may not be known to the traffic intersection controller. Rather, the traffic intersection controller would know the maximum value of the parameter. In the following, we simulate where the parameter $T_{i,j}$ is sampled uniformly from a distribution $[2,r_1]$. 

The traffic-intersection controller informs the time to the AV where it is certain that there is no vehicle from the conflicting lanes. Specifically, the traffic-intersection controller computes the trajectories, the entry and exit times (eqns. (20)-(23)) by assuming that the head-way for the HDVs as $r_1$ seconds (the maximum value). Fig.~\ref{fig:worstdelay_r} and \ref{fig:meandelay_r} show that as $r_1$ increases the delay increases (both the worst-case and the mean value). The increase in delay is also prominent as the percentage of the HDV increases which is expected. Also note that when $r$ is high and the penetration of the AV is small, the number of HDVs with a higher head way distance increases, hence, the delay increases significantly. 

\subsection{Impact of the values of $v_{m}$}\label{sec:diff_speed}
We, empirically, show the impact of $v_m=v_{max}$. As $v_{max}$ decreases the delay increases. The rate of increase of delay increases as the maximum velocity $v_{max}$ decreases (Fig.~\ref{fig:delvsvmax}). 
 \begin{figure*}
     \centering
     \begin{minipage}{0.32\linewidth}
     \includegraphics[width=0.98\textwidth]{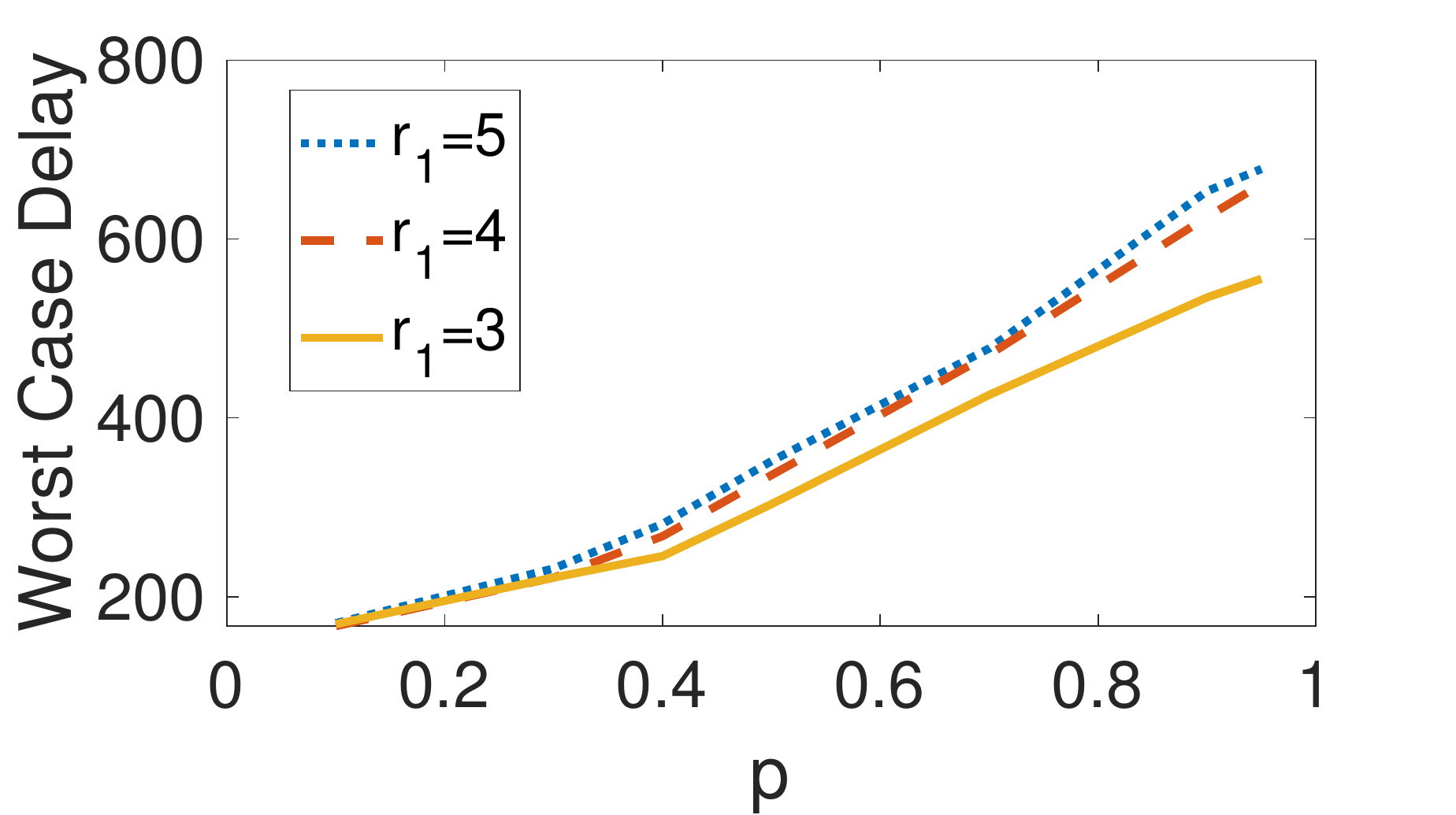}
     \caption{The variation of worst-case delay with respect to $r_1$, $T_{i,j}\sim \mathcal{U}[2,r_1]$ for HDVs.}
     \label{fig:worstdelay_r}
     \end{minipage}\hfill
     \begin{minipage}{0.32\linewidth}
      \includegraphics[width=0.98\textwidth]{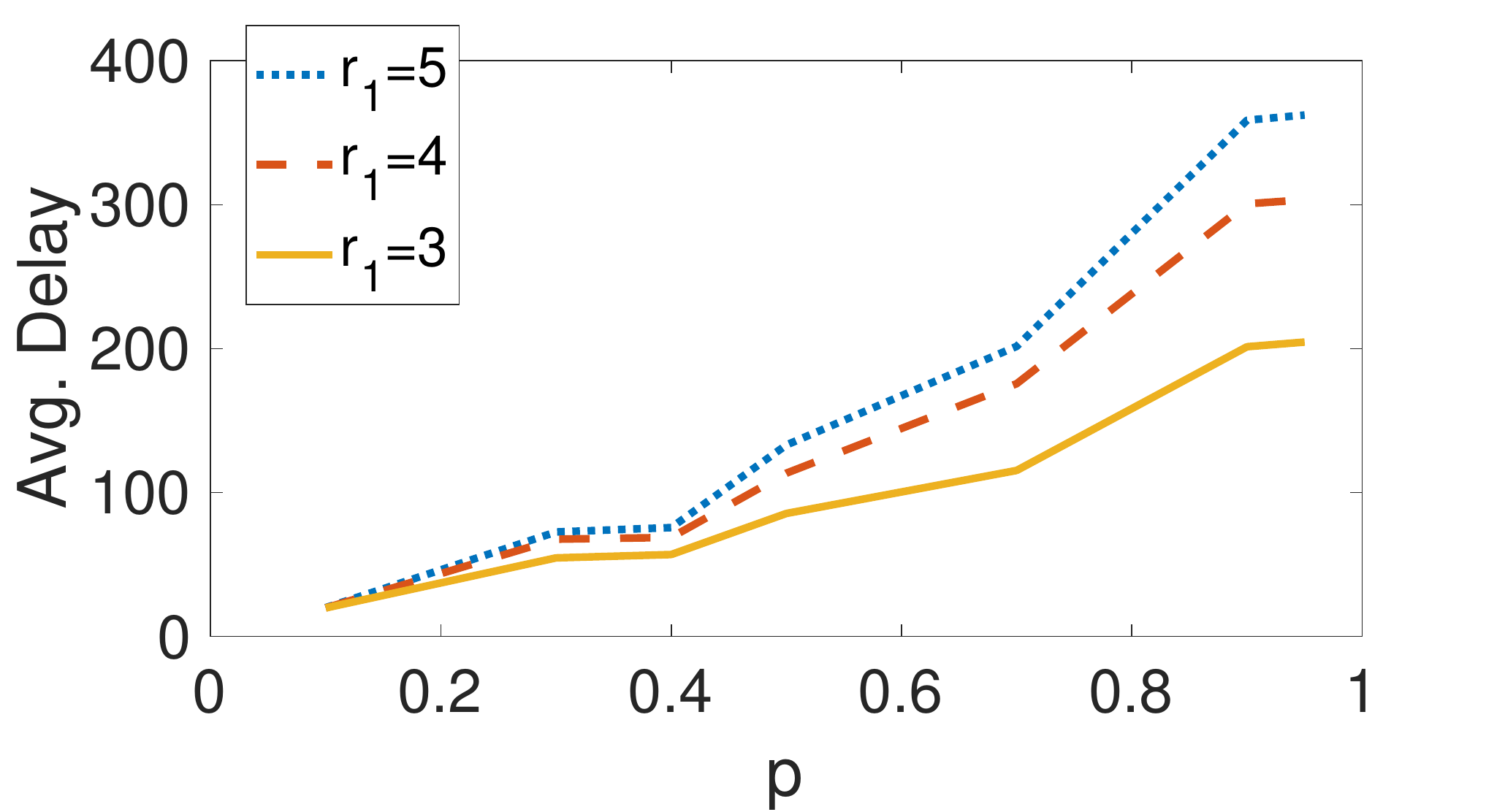}
      \caption{The variation of the average delay with $r_1$.}
      \label{fig:meandelay_r}
     \end{minipage}\hfill
     \begin{minipage}{0.32\linewidth}
      \includegraphics[width=0.98\textwidth]{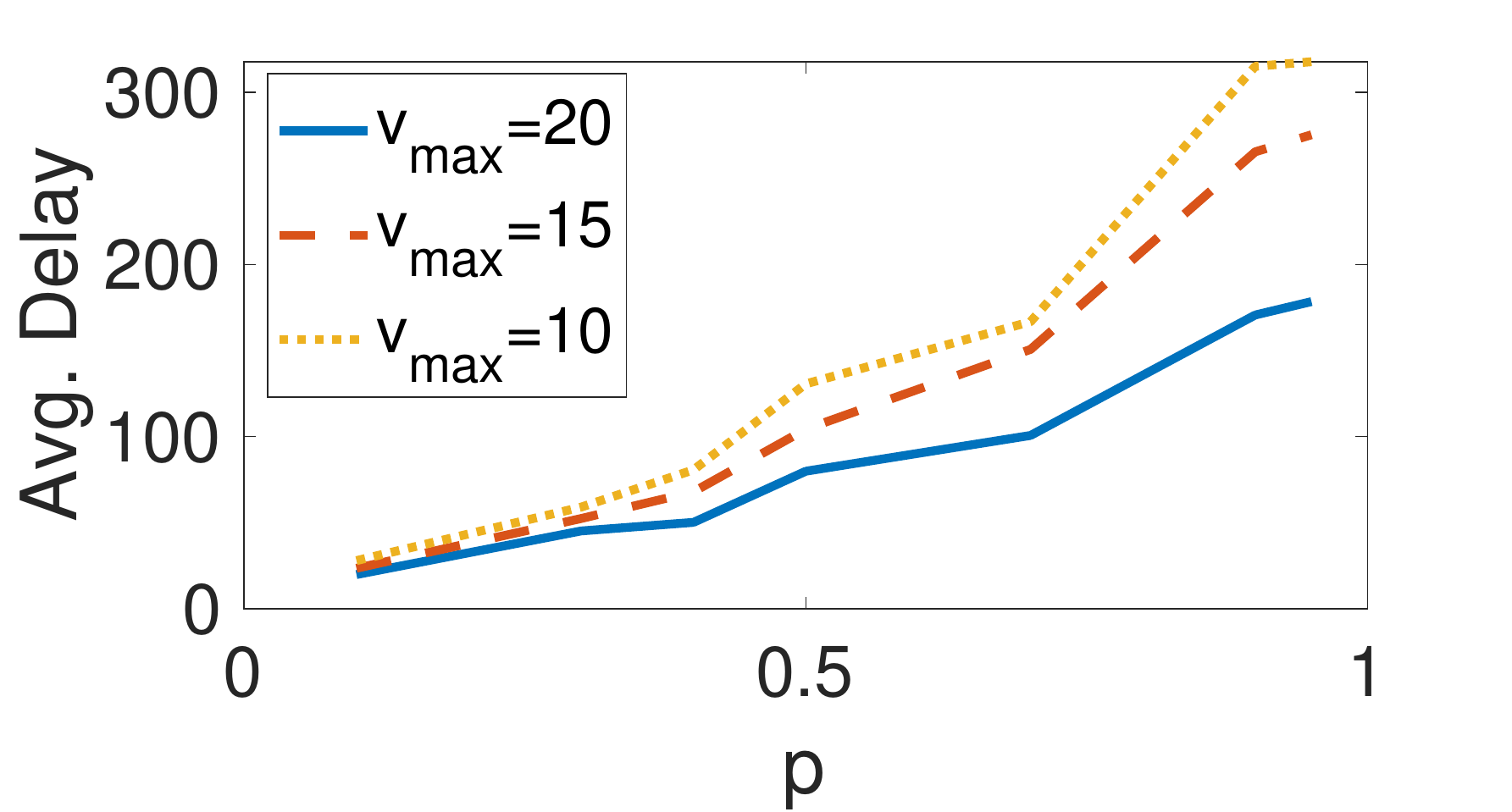}
      \caption{The variation of the average delay with $v_{max}$.}
      \label{fig:delvsvmax}
     \end{minipage}
 \end{figure*}
 \subsection{Velocity Profiles of vehicles}
 In Figs.~\ref{fig:vel_av} and \ref{fig:vel_hdv}, we plot the velocity profile of an AV and a HDV when the AV penetration rate is 50\%. When the AV becomes the lead-vehicle it faces a red-light, hence, the traffic-intersection controller informs the AV at what time it can enter the intersection. Note from Fig.~\ref{fig:vel_av} that the AV enters the intersection at the maximum speed and maintains that till it exits the intersection. 
 
 On the other hand, the traffic-intersection controller can not coordinate with the HDV. The HDV stops twice (Fig.~\ref{fig:vel_hdv}).. First, when the vehicle it is follwing has stopped; next, when the HDV is a lead-vehicle and it faces a red-light. Note that when the traffic-light is again green, it enters the intersection at the $0$ velocity unlike the AV. Hence, the HDV accelerates and reaches the maximum velocity before exiting the intersection. The delay is higher for the HDV compared to the AV. 
 \begin{figure}
     \centering
     \includegraphics[width=50mm]{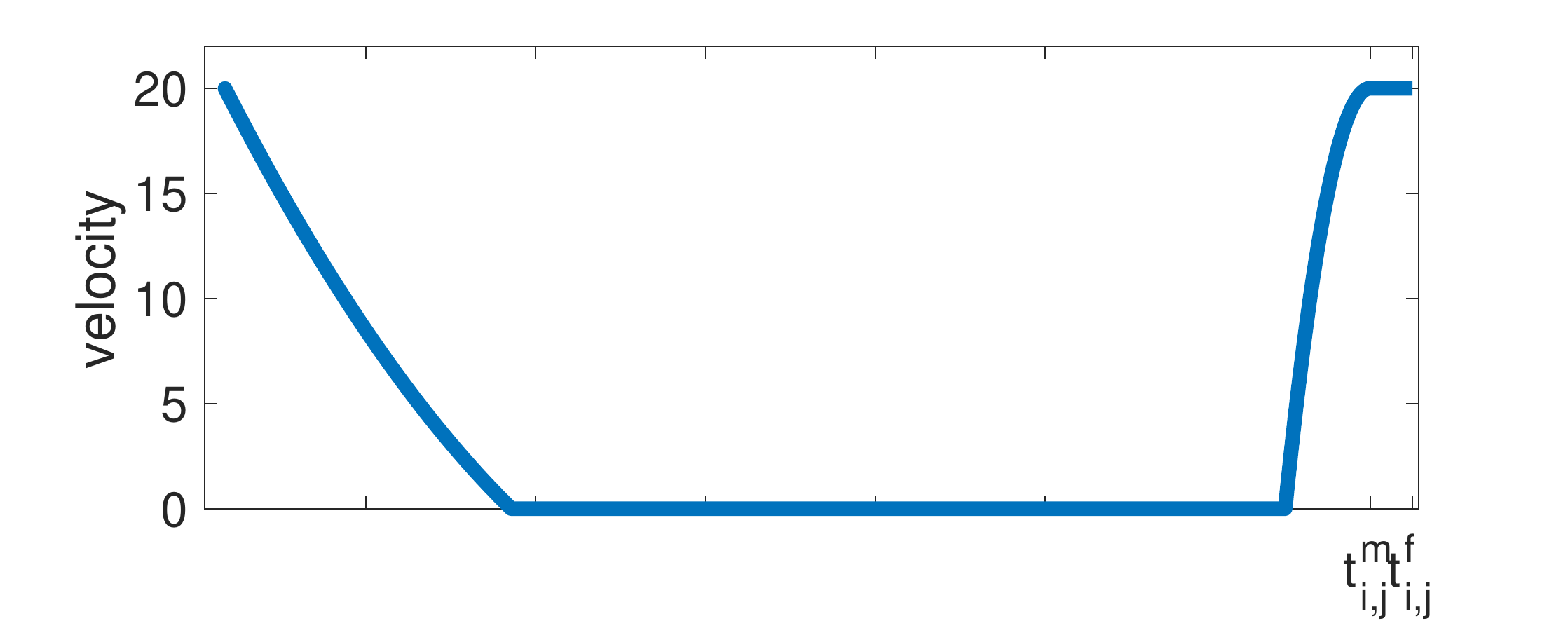}
     \caption{The variation of the velocity of an AV with time.}
     \label{fig:vel_av}
     \vspace{-0.2in}
 \end{figure}
\end{color}
\begin{figure}
    \centering
    \includegraphics[width=50mm]{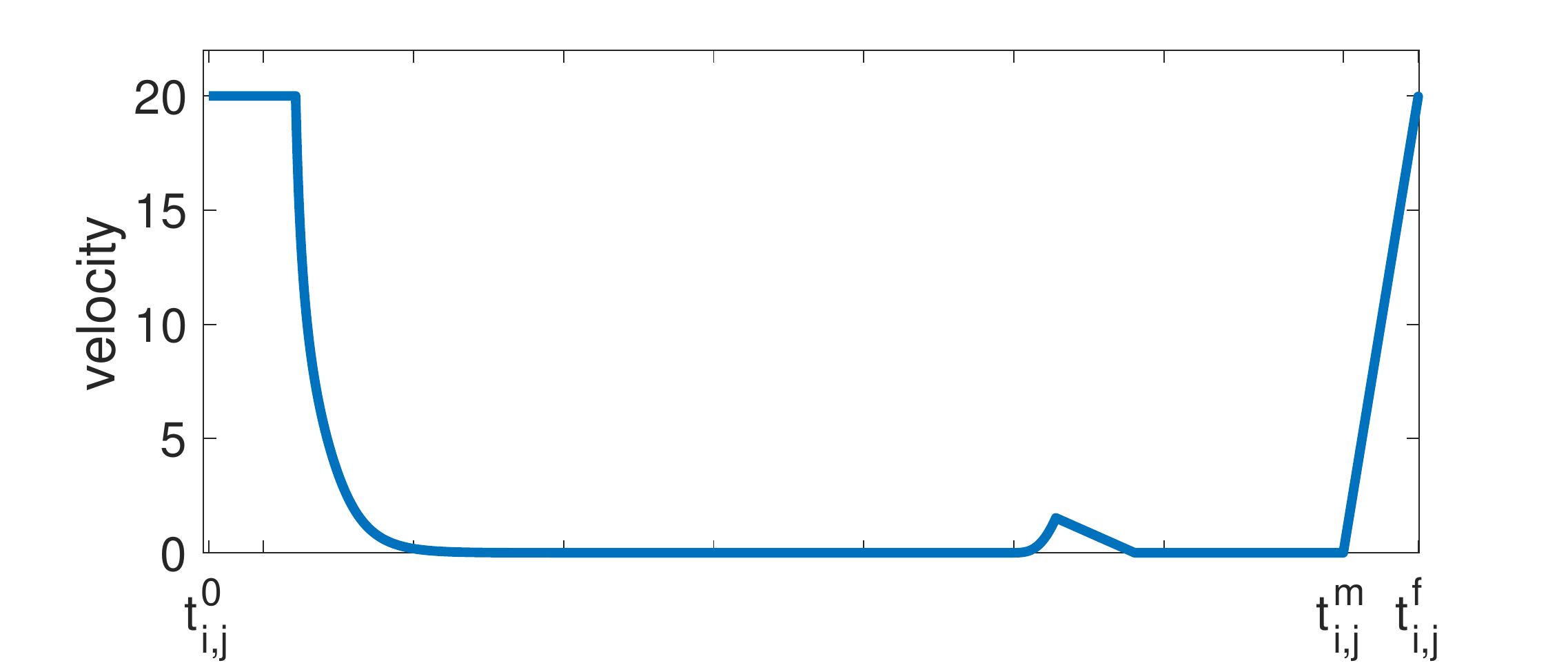}
    \caption{The variation of the velocity of the HDV with time.}
    \label{fig:vel_hdv}
    \vspace{-0.25in}
\end{figure}
\subsection{Arrangement of the HDVs and the AVs}
In order to analyze the impact of the arrangement of the HDVs and AVs, we simulate the arrivals of the vehicles in the following manner. We, first, consider that 50\% of all vehicles is AV. We divide the total incoming vehicles in a lane in the batch of $k$ vehicles. We consider that out of $k$ vehicles, first $k/2$ vehicles are AV and the rest $k/2$ vehicles are HDV.

 Our result shows that as $k$ increases, the total time the vehicles take to cross the intersection increases (Fig.~\ref{delay_permu}). Intuitively, when the HDVs become the lead vehicles they enter the intersection at a smaller speed compared to the AVs. Since the HDVs arrive later compared to the AVs, it takes more time to cross the intersection. 

Fig.~\ref{delay_permu} also shows that the second combination where first $k/2$ vehicles are HDVs and the rest of the $k/2$ vehicles are AVs gives smaller total time compared to the first scenario. Further, as $k$ increases the total time taken is also reduced. When the AV becomes the lead vehicle it can approach the intersection with a larger speed and thus, it takes smaller time to cross the intersection. Thus, the traffic controller can intentionally allow the HDV to cross the intersection by withholding the AVs. However, when $k$ increases too much the total time again increases since the HDVs need to be stopped as otherwise it would cause significant wait time to all the vehicles. Thus, the total time taken again increases when $k$ exceeds a threshold. 
\begin{figure}
\begin{center}
\includegraphics[width=0.3\textwidth]{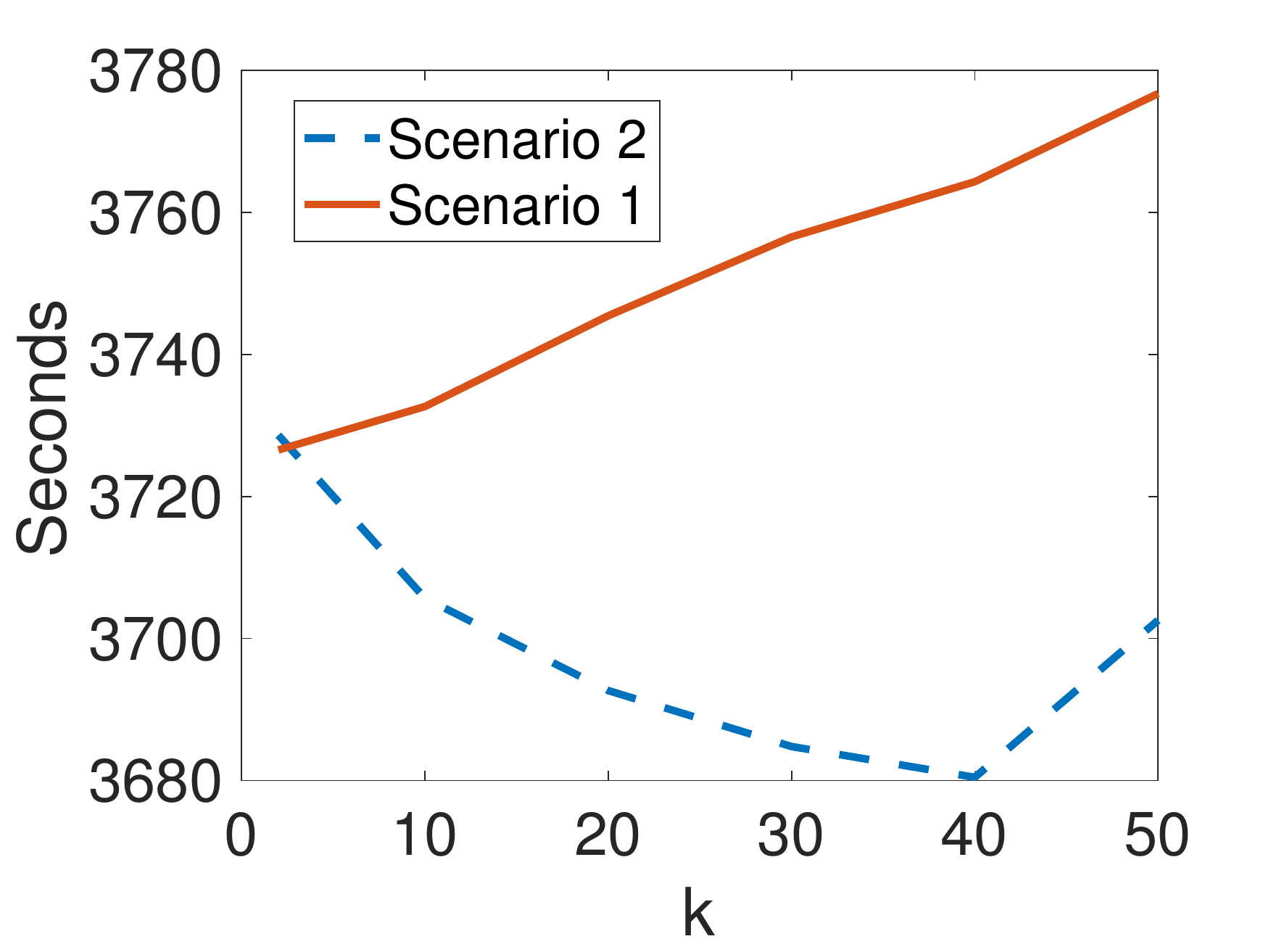}
\vspace{-0.1in}
\caption{The variation of the total time taken by the vehicles in 2 different scenarios as a function of $k$. $\lambda_1=100,\lambda_2=100,\lambda_3=100$. }
\label{delay_permu}
\end{center}
\vspace{-0.3in}
\end{figure}

Fig.~\ref{wait_permu} shows the maximum delay among all the vehicles as a function of $k$. Similar to Fig.~\ref{delay_permu} Scenario 2 reduces the maximum wait time of the vehicles. However, unlike in Fig.~\ref{delay_permu}, in Fig.~\ref{wait_permu} the wait time increases as $k$ increases. This is due to the fact that as $k$ increases, the number of HDVs in a slot becomes higher compared to the AVs. Since the HDVs take longer time to cross the intersection, the wait time of the vehicles become large. Hence, for better performance the HDV and AV should alternate. 
\begin{figure}
\begin{center}
\includegraphics[width=0.3\textwidth]{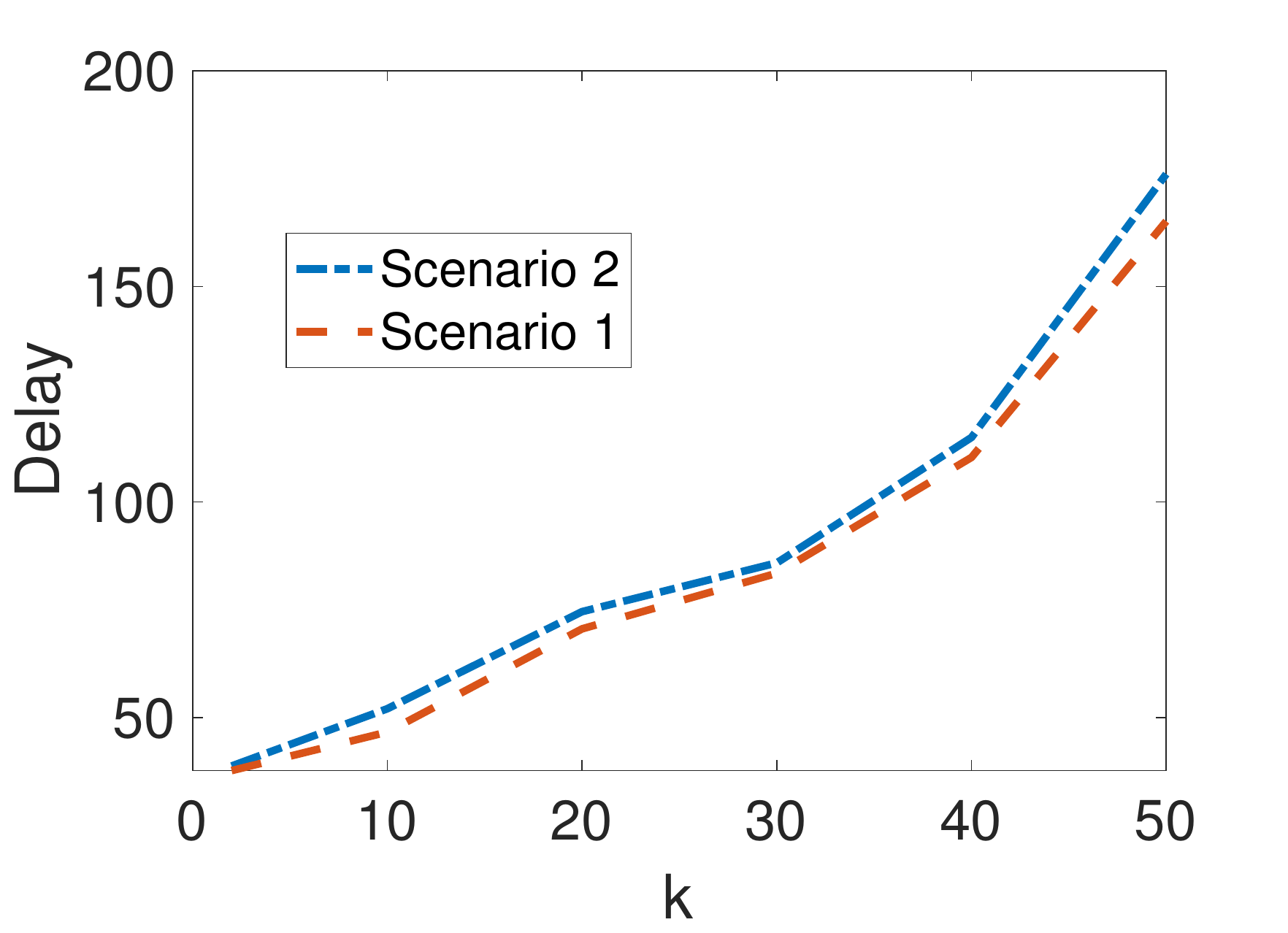}
\vspace{-0.05in}
\caption{The variation of the maximum wait time among the vehicles 2 different scenarios as a function of $k$. $\lambda_1=100,\lambda_2=100, \lambda_3=100$. }
\label{wait_permu}
\end{center}
\vspace{-0.3in}
\end{figure}

 Fig.~\ref{vel_permu} shows the mean velocity of the vehicles when they enter the intersection as a function of $k$. It reveals that scenario 2 gives slightly better mean velocity compared to the scenario 1. However, unlike in Fig.~\ref{delay_permu}, Fig.~\ref{vel_permu} reveals that as $k$ increases the mean velocity decreases in both the scenarios. Since when a large number of HDVs arrive in burst without any AV, it would decrease the mean speed. Hence in order to maintain high velocity the HDV and AV should alternate. 
\begin{figure}
\begin{center}
\includegraphics[width=0.3\textwidth]{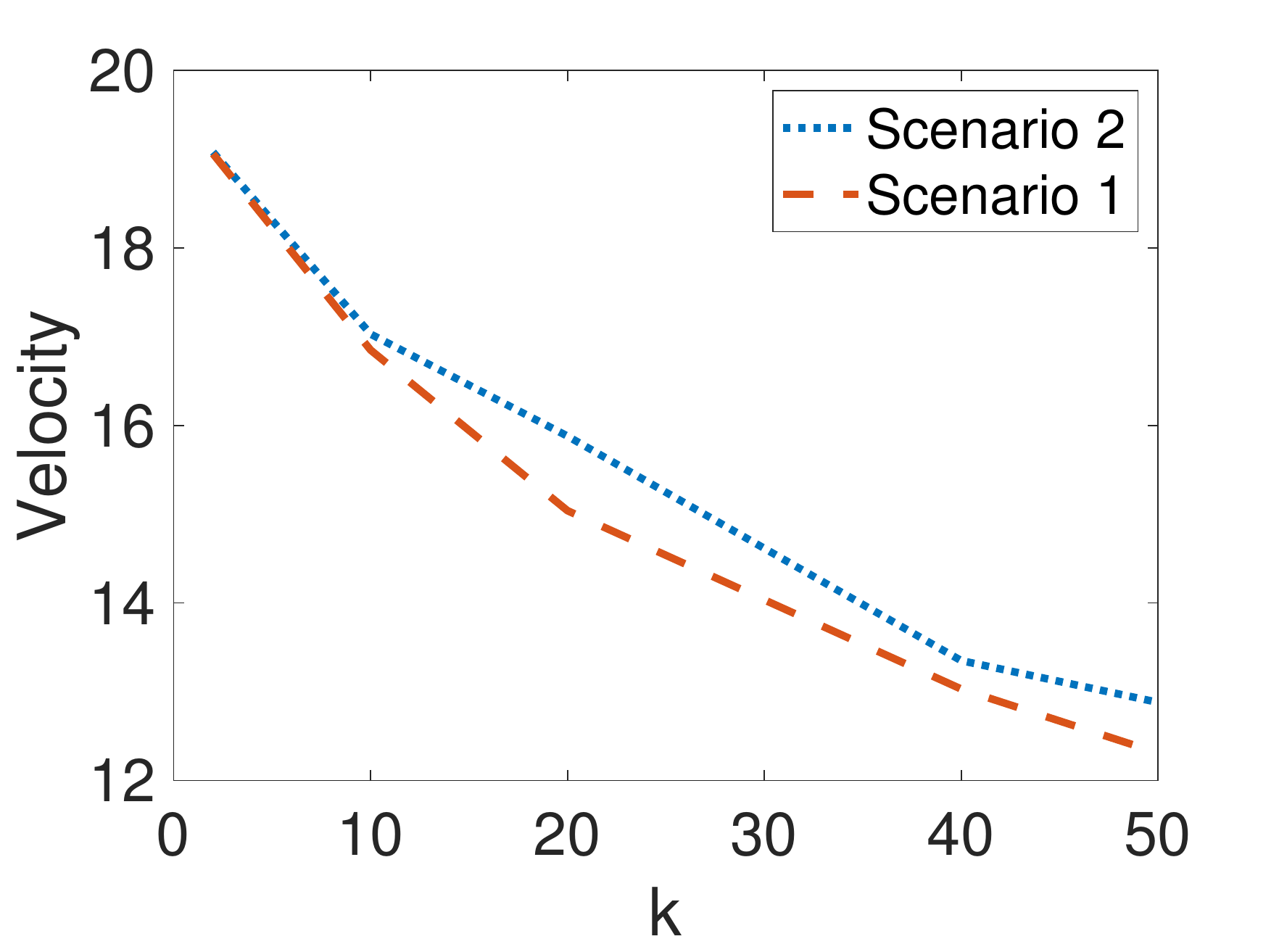}
\vspace{-0.1in}
\caption{The variation of the mean velocity of the vehicles in 2 different scenarios as a function of $k$. $\lambda_1=100,\lambda_2=100, \lambda_3=100$. }
\label{vel_permu}
\end{center}
\vspace{-0.3in}
\end{figure}

Thus, the arrangement of the HDVs and AVs can play a significant role in the total time taken by the vehicles, the average velocity, and the maximum delay among the vehicles. Scenario 2 where the AVs would follow the HDV would give a lower total time for crossing across the vehicles. However, there exists an optimal threshold for the size of the AVs to follow the HDV. On the other hand, when AV and HDV alternate the average velocity at which the vehicles enter is the maximum. The combination where the AV and HDV alternates also results in smaller delay across the vehicles. 

\section{Conclusions and Future Work}
We consider a scenario where the traffic controller  optimally decides the sequence at which vehicle enter an urban intersection in a mixed autonomy scenario. When a new vehicle enters the system, the traffic-controller computes the optimal sequence. Since the number of possible combinations increases exponentially with the number of vehicles, we consider a relaxed version where the sequence is computed only among the newly entered vehicle and the vehicles at the conflicting lanes with the relative order of vehicles among the others are kept the same. Thus, the possible combination scales only linearly with the number of vehicles. We formulate the optimization problem. We consider that the traffic-intersection controller can inform the AV the time at which it can enter the intersection if it enters as the lead vehicle or it will face a red-light once the preceding vehicle enters the intersection. Such an AV can enter the intersection at the maximum speed while adapting its dynamics. We characterize how the traffic-intersection controller selects the traffic-lights at each lane based on a scheduling order, and how such a decision impacts the dynamics of the vehicles. We propose an algorithm which determines the optimal sequence by evaluating the objective for each possible sequence. We also show that how our work can be extended when the HDV's parameters may deviate from the nominal values. We empirically evaluate our algorithm and show that our algorithm significantly outperforms the FIFO mechanism. 

\begin{color}{black}
Our work can be extended in numerous ways. The characterization of optimal algorithm where vehicles are allowed to  make turns is left for the future. We have considered only one intersection zone. The characterization of optimal algorithm considering multiple intersection zones and the traffic behavior across multiple intersection zones is another possible future research direction. Further, though, we empirically characterize the impact of the uncertainties in the predicting the position on the model, the complete theoretical analysis is left for the future. The characterization on other vehicle following models have been left for the future. Finally, more sophisticated model for fuel cost has also left for the future. \end{color}


\bibliographystyle{IEEEtran}
\bibliography{av}
\appendix
\subsection{Proof of Theorem~\ref{thm:1}}\label{proof}
For ease of exposition, with slight abuse of notation, we remove $j$ from the subscript and we simply denote the $i$-th vehicle in the $j$-th lane as $i$. 

Since $\lambda$ is large, the AV would like to maintain the free velocity when it would enter the intersection. From (\ref{dyna}), the Hamiltonian for vehicle $i$ is
\begin{align}\label{hamil}
H_i(t,p_i(t),v_i(t),u_i(t))=\dfrac{1}{2}u_i^2+\mu_i^pv_i+\mu_i^vu_i
\end{align}
The Euler-Lagrange equation becomes
\begin{align}\label{partp}
\dot{\mu^{p}_i}=-\dfrac{\partial H_i}{\partial p_i}=0 \quad \text{(cf.(\ref{hamil}))}
\end{align}
where $\dot{}$ indicates derivative with respect to time. Further, 
\begin{align}\label{partv}
\dot{\mu_i^v}=-\dfrac{\partial H_i}{\partial v_i}=-\mu_i^p\nonumber\\
\end{align}
The necessary condition for optimality is
\begin{align}\label{partu}
\dfrac{\partial H_i}{\partial u_i}=u_i+\mu_i^v=0
\end{align}
Hence, from (\ref{partp}), $\mu^p_i=a_i$ and from (\ref{partv}) is $\mu^v_i=-a_it-b_i$. Finally, from (\ref{partu}) is 
\begin{align}\label{u}
u_i(t)=a_i(t-t_i^0)+b_i
\end{align}
Note that the vehicle dynamics is computed from the point the vehicle enters the system i.e., $t_i^{0}$. Hence, the all the variables would be retarded by $t_i^{0}$ amount. From the vehicle dynamics in (\ref{dyna}), we obtain
\begin{align}\label{v}
v_i(t)=\dfrac{1}{2}a_i(t-t_i^0)^2+b_i(t-t_i^0)+c_i\nonumber\\
p_i(t)=\dfrac{1}{6}a_i(t-t_i^0)^3+\dfrac{1}{2}b_i(t-t_i^0)^2+c_i(t-t_i^0)+d_i
\end{align}
\begin{color}{black}
Now, the boundary conditions are $p_i(t_i)=0, v_i(0)=v_{free}, p_i(t_i^{m})=L, v_i(t_i^{m})=v_{free}$. Since there are four unknowns, and four boundary conditions the solution can be uniquely determined by solving the system of equations in (\ref{v}). Note that $b_i$is negative and $a_i$ is positive. Thus, the vehicle initially decelerates. The vehicle then gradually reduces the magnitude of deceleration and accelerates after that. 

The above profile is correct if the all the constraints in (\ref{constr}) are satisfied. In the following, we consider the profile where the constraints at some time $t$ are not satisfied. 

Case i: Suppose that $u_i(t)$ becomes less than $u_{min}$ at some time $t=t_1$. Now, the velocity at time $t_1$ is  obtained from (\ref{v}). $u_i(t)$ would be clipped at $u_{min}$. Hence, the vehicle dynamics becomes for $t\geq t_1$.
\begin{align}\label{constacc}
v_i(t)=v_i(t_1)+u_{min}(t-t_1)\nonumber\\
p_i(t)=\dfrac{1}{2}u_{min}(t-t_1)^2+v_i(t_i)(t-t_1)+p_i(t_1)
\end{align}
Now, if $u_i(t)$ is greater than $u_{min}$ at some time in (\ref{v}), we need to find the time the constraint would be inactive. Thus, we need to find the time at which the vehicle dynamics would again become similar to (\ref{v}). Let the time at which the vehicle dynamics would be similar to (\ref{v}) be $t_2$, then
\begin{align}
v_i(t)=\dfrac{1}{2}a_i(t-t_2)^2+b_i(t-t_2)+c_i\nonumber\\
p_i(t)=\dfrac{1}{6}a_i(t-t_2)^3+\dfrac{1}{2}b_i(t-t_2)^2+c_i(t-t_i^0)+d_i\nonumber\\
u_i(t)=a_i(t-t_2)+u_{min}\nonumber\\
c_i=v_i(t_2)\quad d_i=p_i(t_2)
\end{align}
Note that $v_i(t_2)$ and $p_i(t_2)$ is obtained from (\ref{constacc}). Further $v_i(t_i^{m})=v_{free}$ and $p_i(t_i^m)=L$. This would give unique solution for $a_i$ and $t_2$.

Case ii: Suppose that $u_i(t)$ becomes equal to $u_{max}$ at certain time $t_1$. After that $u_i(t)=u_{max}$. Hence, the dynamics becomes for $t\geq t_i$ is
\begin{align}
v_i(t)=v_i(t_1)+u_{max}(t-t_1)\nonumber\\
p_i(t)=\dfrac{1}{2}u_{max}(t-t_1)^2+v_i(t_1)(t-t_1)+p_i(t_1)
\end{align}
Note that since $v_i(t_i^{m})=v_{free}$, and $p_i(t_i^{m})=L$, thus, 
\begin{align}\label{eq:vp}
v_i(t_1)=v_{free}+u_{max}(t_i^{m}-t_1)\nonumber\\
p_i(t_1)=L-\dfrac{1}{2}u_{max}(t-t_1)^2+v_{free}(t-t_1)
\end{align}
Now, before $t_1$ the vehicle dynamics would be similar to (\ref{v}). We need to find optimal $t_1$ such that the boundary conditions are satisfied ,i.e., the conditions in (\ref{eq:vp}) are satisfied. Further, $u_i(t_1)=u_{max}$ this would give unique $a_i, b_i, d_i$.

Case iii: Suppose that $v_i(t)$ becomes $v_{min}$ at some time $t=t_1$. The velocity can not be reduced. If the vehicle dynamics in (\ref{v}) mandates that the velocity must go below $v_{min}$, the vehicle $i$ can only move at velocity $v_{min}$. Before time $t_1$, the dynamics is given by in (\ref{v}). After that $u_i(t)=0$, $v_i(t)=v_{min}, p_i(t)=v_{min}(t-t_1)+p_i(t_1)$ where $p_i(t_1)$ is given by the equation in (\ref{v}). If the vehicle dynamics in (\ref{v}) mandates that the velocity $v_i(t)$ must be greater than $v_{min}$ at some time. The vehicle dynamics then becomes after time $t_2$
\begin{align}\label{ui0}
u_i(t)=a(t-t_2)\nonumber\\
v_i(t)=\dfrac{1}{2}a(t-t_2)^2+v_{min}\nonumber\\
p_i(t)=\dfrac{1}{6}a(t-t_2)^3+v_{min}(t-t_2)
\end{align}
with the boundary condition that $v_i(t_i^{m})=v_{free}$, $p_i(t_i^{m})=L$, $p_i(t_2)=p_i(t_1)+v_{min}(t_2-t_1)$, $a>0$. The above set of equations give unique values of $a$ and $t_2$. 

Case iv: Suppose that $v_{i}(t)$ becomes equal to $v_{min}$ while $u_i(t)=u_{min}$ at time $t_2$, then the vehicle dynamics would be $u_i(t)=0$, $v_i(t)=v_{min}$ after time $t_2$. If the velocity requires to be greater than $v_{min}$, the vehicle dynamics becomes similar to (\ref{ui0}) after time $t_3$.

Case v: The original solution may give rise the scenario where $p_i(t)>s$ for some $s$, and velocity would be negative in order to make $p_i(t_i^m)=s$. In the above scenario, there must exist a time $t_2$ such that $v_i(t_2)=0$ and $p_i(t_2)<s$. Let in the original solution $s_1=p_i(t_2$). 

There also must exist a time $t_1$ at which $p_i(t_1)=s_1$. Hence, the optimal solution is now given in the form given in (\ref{v}) with the boundary condition that $v_i(t_1)=0$ and $p_i(t_1)=s_1$. Between time $t_1$ and $t_2$ the velocity will be $0$ and the position will be $s_1$. After time $t_2$, the position, velocity and acceleration is given by the values as in the original solution. 
\end{color}
\end{document}